\def\year{2022}\relax

\documentclass[letterpaper]{article} 
\usepackage{aaai22arxiv}  
\usepackage{times}  
\usepackage{helvet}  
\usepackage{courier}  
\usepackage[hyphens]{url}  
\usepackage{graphicx} 
\usepackage{natbib}  
\usepackage{caption} 
\DeclareCaptionStyle{ruled}{labelfont=normalfont,labelsep=colon,strut=off} 
\title{Multi-Leader Congestion Games with an Adversary}
 \author {
     Tobias Harks\textsuperscript{\rm 1},
     Mona Henle\textsuperscript{\rm 2},
     Max Klimm\textsuperscript{\rm 3},
     Jannik Matuschke\textsuperscript{\rm 4},
     Anja Schedel\textsuperscript{\rm 1}
 }

 \affiliations {
     \textsuperscript{\rm 1} University of Augsburg, 86159 Augsburg, Germany\\
     \textsuperscript{\rm 2} University of Applied Sciences Augsburg, 86161 Augsburg, Germany\\
     \textsuperscript{\rm 3} TU Berlin, 10623 Berlin, Germany\\
     \textsuperscript{\rm 4} KU Leuven, 3000 Leuven, Belgium \\
     tobias.harks@math.uni-augsburg.de, mona.henle@hs-augsburg.de, klimm@math.tu-berlin.de, jannik.matuschke@kuleuven.be, anja.schedel@math.uni-augsburg.de
 }
 
\usepackage{amsmath,amsthm,amssymb,amsfonts,amsbsy}
\newtheorem{theorem}{Theorem}
\newtheorem{definition}[theorem]{Definition}
\newtheorem{lemma}[theorem]{Lemma}

\newtheorem{corollary}[theorem]{Corollary}
\newtheorem{example}[theorem]{Example}

\newtheorem*{notation}{Notation} 
\newtheorem{claim}{Claim} 

\newtheorem{applemma}{Lemma}[section]
\newtheorem{appexample}[applemma]{Example}

\newenvironment{proofClaim}[1][]{\ifthenelse{\equal{#1}{}}{\begin{proof}}{\begin{proof}[#1]}}{\end{proof}}
\newcommand{\N}{\mathbb{N}}

\newcommand{\Z}{\mathbb{Z}}
\usepackage[ruled,linesnumbered, norelsize,vlined]{algorithm2e} 
\usepackage{subfig} 
\usepackage{booktabs}
\usepackage{mathtools}
\DeclareMathOperator*{\dev}{dev}
\usepackage{tikz}
\usetikzlibrary{decorations.pathreplacing} 
\usepackage{enumitem}
\usepackage{multirow}
\usepackage{nicefrac}

\nocopyright

\begin{document}

\maketitle
\begin{abstract}
    We study a multi-leader single-follower congestion game where multiple users (leaders)  choose one resource out of a set of resources and, after observing the realized loads, an adversary (single-follower) attacks the resources with maximum loads, causing additional costs for the leaders. For the resulting strategic game among the leaders, we show that pure Nash equilibria may fail to exist and therefore,  we consider approximate equilibria instead.
    As our first main result, we show that the existence of a $K$-approximate equilibrium can always be guaranteed, where $K\approx 1.1974$ is the unique solution of a cubic polynomial equation. To this end, we give a polynomial time combinatorial algorithm which computes a $K$-approximate equilibrium.  
    The factor $K$ is tight, meaning that there is an instance that does not admit an $\alpha$-approximate equilibrium for any $\alpha<K$. Thus $\alpha=K$ is the smallest possible value of $\alpha$ such that the existence of an $\alpha$-approximate equilibrium can be guaranteed for \emph{any} instance of the considered game. 
    Secondly, we focus on approximate equilibria of a given fixed instance. We show how to compute efficiently a best approximate equilibrium, that is, 
    with smallest possible $\alpha$ among all $\alpha$-approximate equilibria of the given instance. 
\end{abstract}
\section{Introduction}
Hierarchical leader-follower games have received considerable attention in the artificial intelligence community, especially, because
several real-world applications related to the protection of vulnerable systems can be modeled within this framework. Applications include the security domain~\citep{Kiekintveld09,MarchesiC019,SinhaFAKT18,GanEW18}, where a leader aims at protecting a set of valuable targets and moves first by applying a
defender strategy such as controls or fortification of resources.  The adversary acts as a follower and, after observing the leader’s defensive strategy, chooses a strategy incurring maximum damage. The leader anticipates
the followers' strategy. Thus, the computation of the defender strategy takes the follower reaction into account.

While most works in this literature consider the case of a single leader, the
case of \emph{multiple leaders} playing a simultaneous-move strategic game subject to one or more followers
has received much less attention and only very few results
with respect to the existence and computational complexity of equilibria are known.
Note that the multi-leader case applies
to several scenarios, for example, in the analysis of deregulated electricity markets in which some of the large energy producers are the leaders and the smaller energy producers and independent system operator are the followers; see~\citet{LeyfferM10} and references therein. Also in the security domain related to transport and communication networks, there are usually multiple leaders that compete over the network resources subject to a followers' response; see~\citet{Kulkarni15}. One well-known obstacle in the analysis of multi-leader games---even in the realm of continuous formulations with convex action spaces for the leaders and single-valuedness of the followers' response---is the inherent non-convexity of the best-response correspondence that results in non-existence of pure Nash equilibria; see~\citet{Kulkarni15}.

In this paper, we consider a class of multi-leader single-follower games on discrete strategy spaces that are motivated by security applications with congestion effects. Consider a standard singleton congestion game where multiple users (leaders) choose one resource out of a set of resources. After observing the realized loads, an adversary (single-follower) attacks the resources with maximum loads, causing additional disutilities for the users on the attacked resources. The adversary may be thought of as either being a malicious player attacking the resources in order to maximize the caused damage or as controls by a central authority to counter tax or fare evasion; see~\citet{CHKM17} for a related mathematical model of fare evasion without any congestion or load balancing effects. In both applications it is sensible to assume that the adversary has limited resources, 
modeled by a fixed budget for his interventions that can be distributed freely on the resources, and that he acts rationally, investing the budget only on resources 
with maximum load. 
The users anticipate this strategy. From their perspective,
every maximum-load resource is equally likely under attack, that is, they assume the budget
to be spent evenly among the resources with maximum load (this
can be interpreted as a randomized strategy of the adversary choosing the uniform distribution 
over maximum-load resources). For the users, the additional cost
term corresponds to the expected additional
damage cost due to an attack.

This fundamental model has, to the best of our knowledge, not been analyzed
before and we investigate the existence and computation of (approximate) pure Nash equilibria of this multi-leader single-follower game.

\subsection{Our Results and Proof Techniques}
We first observe that pure Nash equilibria do not always exist in the introduced game, not even for linear congestion costs. 
This motivates the analysis of approximate pure Nash equilibria, where any unilateral deviation cannot improve the cost of the deviating leader by more than a factor $\alpha$, for some $\alpha\geq 1$. (Note that the adversary is still assumed to act optimally.) We analyze existence and efficient computation of approximate equilibria for the introduced game with linear congestion costs. 

As our first main result, we show that $\alpha=K\approx 1.1974$ is the smallest possible value of $\alpha$ such that the existence of an $\alpha$-approximate pure Nash equilibrium can always be guaranteed ($K$ is the unique solution of some cubic polynomial equation). 
For the proof, we give an efficient algorithm which computes a $K$-approximate equilibrium. 
The basic approach is to start with an empty game, and add the players one after another, always placing them on a best-response resource. If the addition of a new player makes some of the already added players ``unhappy'', meaning that there is a unilateral deviation decreasing their cost by more than a factor $\alpha$, we let the unhappy players deviate one after another until all players are happy again, that is, an $\alpha$-approximate pure Nash equilibrium is reached for the subset of players already added. Only then the next player is added. 
By choosing the possible deviations carefully, we can show that this procedure terminates after a polynomial number of steps for $\alpha=K$, showing existence of $K$-approximate pure Nash equilibria and giving an efficient way of computing them.
A similar approach has been used before to compute \emph{exact} pure Nash equilibria in the context of weighted congestion games~\citep{Milchtaich96,Ackermann09}, but we are not aware of any results regarding approximate equilibria applying this technique.
We furthermore provide an instance which does not admit an $\alpha$-approximate pure Nash equilibrium for any $\alpha<K$. 
This shows that $\alpha=K$ is tight in the sense that it is the smallest possible value such that the existence of an $\alpha$-approximate equilibrium can be guaranteed for \emph{all} instances of the introduced game. 

However, for a single given instance, better approximate equilibria might exist, that is, there may be $\alpha$-approximate equilibria with $\alpha<K$. 
We show how to compute efficiently, for a given instance, a \emph{best} approximate equilibrium,
that is, an $\alpha$-approximate equilibrium for the smallest value of $\alpha$ for which such an equilibrium exists.
Note that this in particular implies that we can decide efficiently whether a given instance admits an exact pure Nash equilibrium, and in case of existence, we can also compute such an equilibrium. 
Our algorithm is based on a careful analysis of the structure of optimal approximate equilibria, 
which allows us to enumerate a polynomially-sized set of possible resource-load configurations, from which an optimal approximate equilibrium can then be found using a simple linear program.

\subsection{Related Work}
The game that we analyze in this paper constitutes a Stackelberg game with multiple leaders and a single follower. The leaders' game is a singleton congestion game and we assume symmetric strategies, meaning that all leaders have the same strategy space. 
Stackelberg games with an underlying congestion game for (a subset of) the players have received considerable attention in the literature. 
\citet{CastiglioniMGC19} and~\citet{MarchesiC019}  consider a game with a \emph{single} leader and multiple followers where all players participate in a congestion game (but the leader's congestion cost functions may be different from the followers'). Depending on the structure of strategy spaces and congestion cost functions, they analyze the computational complexity of computing exact equilibria. In particular, they find that efficient algorithms are only possible for singleton strategy spaces (unless $P=NP$), and derive such algorithms for singleton strategy spaces where either all followers have the same strategies~\citep{CastiglioniMGC19}, or the followers can be divided in ``classes'' having the same strategies~\citep{MarchesiC019}.

There are several works analyzing hierarchical situations with a subsequent nonatomic network routing game, where a set of infinitesimally small players chooses paths in a network, and each player aims to minimize the (load-dependent) length of her chosen path leading to a Wardrop equilibrium~\cite{Wardrop52}. 
For works analyzing situations where a single leader determines capacities or prices in order to reduce the total congestion (plus investments for the case of capacities) of the Wardrop equilibria in a subsequent network routing game, we refer to \citet{Marcotte85mp,GairingHK17} for setting capacities, and~\citet{Beckmann56} and \citet{Yang04} for setting prices. \citet{Labbe98} study a model where a single leader sets prices in order to maximize her profit in a subsequent network routing game (but without congestion effects). 
\citet{HarksSV19} and \citet{CorreaGLNS18} consider a game where multiple leaders set prices in order to maximize their own profits achieved in a network routing game. The prices that the leaders are allowed to choose are upper-bounded by price caps (leader-specific in~\citet{CorreaGLNS18}, equal for all leaders in~\citet{HarksSV19}), and the two papers consider the (three-level) problem of a system designer who chooses the cap(s) in order to minimize total congestion. Finally, models where multiple leaders choose prices and capacities to maximize their individual profits achieved in a network routing game are for example analyzed by~\citet{JohariWR10}, \citet{Liu2011}, and \citet{HarksSchedel19}.

Regarding the computation of approximate equilibria in atomic congestion games, we refer to~\citet{CaragiannisFGS11,CaragiannisFGS15}. Finally, we also mention here  congestion games
with an adversarial structure such as agent  or resource failures, see~\citet{BiloMV18,MeirTBK12,Yupeng17}  or games with malicious players~\cite{BabaioffKP09}. 

\begin{table*}[tb]%
\centering
\begin{tabular}{ccc}
\toprule
load profile & deviation (of some player using $r$ to $r'$, notation $r \rightarrow r'$) & resulting cost improvement\\ \midrule
$(5,0,0)$ & $r_1 \rightarrow r_2$ & $5a_{r_1}+B\phantom{/2}=6 > 2= \phantom{2}a_{r_2} \phantom{{}+{}B}$ \\
$(4,1,0)$ & $r_1 \rightarrow r_2$ & $4a_{r_1}+B\phantom{/2}=6 >4=2a_{r_2}\phantom{{}+{}B}$  \\
$(3,2,0)$ & $r_1 \rightarrow r_3$ & $3a_{r_1}+B\phantom{/2}=6>5=\phantom{2}a_{r_3}\phantom{{}+{}B}$ \\
$(3,1,1)$ & $r_3 \rightarrow r_2$ & $\phantom{3}a_{r_3} \phantom{ {}+{} B/2}=5> 4=2a_{r_2}\phantom{{}+{}B}$  \\
$(2,2,1)$ & $r_2 \rightarrow r_1$ & $2a_{r_2}+B/2=7>6= 3a_{r_1}+B$ \\\bottomrule
\end{tabular}
\caption{Improving deviations for the candidate profiles for the game in Example~\ref{example_noPNE}.}
\label{table_example_noPNE}
\end{table*}

\section{The Model}\label{sec:model}
For an integer $k \in \Z_{\geq 0}$, let $[k] := \{1,\dots,k\}$. 
 Let $N = [n]$ be a finite set of players (leaders) and $R = \{r_1,\ldots,r_m\}$ be a finite set of $m$ resources. For each player~$i$, the set of strategies available to player~$i$ is $X_i=R$. 
We call $x = (x_1,\dots,x_n)$ with $x_i \in X_i$ for all $i \in N$ a \emph{strategy profile}, and $X = X_1 \times \dots \times X_n$ the \emph{strategy space}.

We use standard game theory notation; for a strategy profile $x \in X$, we write $x = (x_i, x_{-i})$ meaning that $x_i$ is the strategy that player~$i$ plays in $x$ and ${x}_{-i}$ is the partial strategy of all players except $i$. 
Every strategy profile $x = (x_1,\dots, x_n) \in X$ induces a \emph{load} or \emph{congestion} on the resources given by
$ \ell_r(x):=|\{i \in N\mid  x_i=r\}|, r\in R.$
 We are further given linear cost functions $a_r\ell_r(x) ,r\in R$ with nonnegative coefficients $a_r\geq 0$.
In classical congestion games, 
 the private cost of player~$i$ under strategy profile $x \in X$ is defined as $\pi_i(x) = a_{x_i}\ell_{x_i}(x).$ 
Now we model the actions of an \emph{adversary} (follower) after the leaders have
chosen their joint strategy profile $x\in X$.
Formally, given $x\in X$ the adversary solves
\begin{equation}\tag{$LP$}\label{lp-follower}
\max\sum_{r\in R} \ell_r(x)\kappa_r\;\;
\text{s.t.:}\quad  \sum_{r\in R}\kappa_r \leq B, \; \kappa\geq 0.
\end{equation}
The linear program~\eqref{lp-follower} has the interpretation that the adversary has a budget of $B > 0$ that can be freely distributed among the resources. For each unit of budget spent on a resource, the adversary receives a utility equal to the number of players on that resource since any interaction with a leader on a resource is equally beneficial for the follower. Thus, the adversary
strategically selects
those resources that are used by the 
maximum number of players in order to maximize the caused damage. It is not hard to see that these are precisely the optimal solutions to~\eqref{lp-follower}. While~\eqref{lp-follower}
may have multiple optimal solutions,
a reasonable selection among the optimal
solutions is the following, where we use the notation $M(x):=\max_{r\in R}\{\ell_r(x)\}$ together with 
$M^{-1}(x):=\arg\max_{r\in R}\{\ell_r(x)\}$:
\begin{equation}\label{eq:attack}\kappa_r^*(x)=\begin{cases} \frac{B}{|M^{-1}(x)|}, & \text{ if } r \in M^{-1}(x), \\ 0, &\text{ else.}\end{cases}\end{equation}
Clearly, $\kappa_r^*(x)$ is an optimal solution to~\eqref{lp-follower} and has the intuitive interpretation 
that, assuming that every maximum-load resource is equally likely to be under attack by the adversary,  from the perspective of the players it represents the expected additional
resource cost due to an attack.

The multi-leader congestion game with an adversary  can be defined as the  game $G=(N,X,B,\pi)$ in strategic form,
where
\begin{equation}\label{eq:cost_attack}
\pi_i(x):= a_{x_i}\ell_{x_i}(x)+\kappa_{x_i}^*(x).\end{equation}
We furthermore define
\begin{equation}
    c_r(x):=a_r\ell_r(x)+\kappa_{r}^*(x)
\end{equation}
for all $r\in R$  (this is useful in case we do not want to consider a specific player using $r$). 

A strategy profile $x\in X$ is called a \emph{pure Nash equilibrium (PNE)} of $G$ if for all $i\in N$:
\[ \pi_i(x)\leq \pi_i(y_i,x_{-i}) \text{ for all }y_i\in X_i.\]

Let us give an example of a multi-leader congestion game with an adversary showing
that pure Nash equilibria need not exist in general.

\begin{algorithm*}[tb]
\caption{Computation of an $\alpha$-approximate PNE.}\label{alg:approxnash}
\SetVlineSkip{0.3em} 
\SetNlSkip{-0.75em}
\SetInd{1.6em}{0.7em}
\SetKwInput{Input}{Input}
\SetKwInput{Output}{Output}
\Indentp{1em}
\Input{Player set $N = [n]$, resource set $R=\{r_1,\ldots,r_m\}$, resource cost coefficients $0\leq a_{1}\leq \cdots \leq a_{m}$, and $\alpha\geq 1$.}
\Output{$\alpha$-approximate pure Nash equilibrium $x$.}

$x \leftarrow (0,0,\ldots,0)$; \\
\For{$k=1,2, \ldots, n$}{\label{for-loop}
$k'\leftarrow \min\{i\in[m]: c_{r_{i}}(x_{-k},r_{i}) \leq c_r(x_{-k},r)$ $\forall r \in R\}$; \label{newplayer_br1}\\
$x_k \leftarrow r_{k'}$;\label{newplayer_br2}\\
\While{$W_{\alpha}(x) \ne \emptyset$}{ \label{while}
$R_{\alpha}(x)\leftarrow \{r\in R: x_i=r \text{ for some }i\in W_{\alpha}(x)\}$; \\
$i^-\leftarrow \max\{j\in [m]: r_{j} \in R_{\alpha}(x) \text{ and } c_{r_{j}}(x)\geq c_{r}(x) \ \forall r\in R_{\alpha}(x)\}$; \label{while1}\\
$i\leftarrow$ some player with $x_i=r_{i^-}$;  \label{while2}\\
$i^+ \leftarrow \min\{j\in [m]: c_{r_{j}}(x_{-i},r_{j}) \leq  c_r(x_{-i},r) \ \forall r \in R\}$; \label{while3} \\
 $x_i \leftarrow r_{i^+}$;
}
}
\end{algorithm*}

\begin{example}\label{example_noPNE}
Consider the game with $m=3$ resources, $n=5$ players, budget $B=6$ and resource cost coefficients 
\[
a_{r_1}=0, \ a_{r_2}=2, \text{ and } a_{r_3}=5.
\]
We proceed to show that this game has no pure Nash equilibrium. 
To this end, we show for each strategy profile~$x$ that there exists a player who can decrease her cost by a unilateral deviation from $x$. 
Since the players are symmetric, it suffices to analyze all ``load'' profiles $(\ell_{r_1}(x),\ell_{r_2}(x),\ell_{r_3}(x))$, that is, all vectors $(\ell_{r_1},\ell_{r_2},\ell_{r_3})\in \N^3$ with $\ell_{r_1}+\ell_{r_2}+\ell_{r_3} =n=5$. 
Since $a_{r_1}<a_{r_2}<a_{r_3}$, the only candidates for a PNE are those strategy profiles $x$ with $\ell_{r_1}(x)\geq\ell_{r_2}(x)\geq\ell_{r_3}(x)$. 
Thus, it suffices to show that there is no PNE among the five load profiles $(5,0,0), (4,1,0), (3,2,0), (3,1,1)$ and $(2,2,1)$. 
Table~\ref{table_example_noPNE} provides an improving deviation for each of these candidate profiles, showing that the game has no PNE. 
\end{example}
This example motivates the analysis of approximate equilibria defined as follows. 
\begin{definition}
A strategy profile $x\in X$ is an  $\alpha$-approximate pure Nash equilibrium ($\alpha$-PNE) of $G$ for some $\alpha\geq 1$, if for all $i\in N$:
\[ \pi_i(x)\leq \alpha \cdot \pi_i(y_i,x_{-i}) \text{ for all }y_i\in X_i.\]
A unilateral deviation which decreases the cost of the deviating player more than a factor $\alpha$ is called an $\alpha$-improving deviation, or an $\alpha$-improving move.
\end{definition}
For $\alpha=1$, we obtain the standard PNE.
For general $\alpha\geq 1$, the interpretation is that no player can improve her cost by a unilateral deviation gaining more than a factor $\alpha$.
We remark that, while one can similarly define \emph{additively} approximate equilibria, no existence guarantees can be given for such equilibria for any additive constant due to the scale-invariance of the games studied in this article (see Appendix~\ref{app:additive-approximate} for details).

\section{Computing $K$-approximate PNE}\label{sec:unrestricted}

In this section, we analyze approximate PNE of the introduced multi-leader congestion games with an adversary. As Example~\ref{example_noPNE} shows, existence of (exact) PNE can not be guaranteed for these games.
We show that $\alpha=K\approx 1.1974$ is the smallest possible $\alpha$ such that the existence of an $\alpha$-PNE can be guaranteed for any instance, where
\begin{align}
\label{def:K}
K:= \frac{1}{6} \biggl(1+\sqrt[3]{109-6\sqrt{330}}+\sqrt[3]{109+6\sqrt{330}}\,\biggr)
\end{align}
is the unique solution of the equation $-x^3+x^2/2+1=0$. 
To this end, we provide Algorithm~\ref{alg:approxnash} which efficiently computes a $K$-approximate PNE (see Theorem~\ref{theo_symm_and_linear_1}). This result is complemented by an instance where no $\alpha$-approximate PNE with $\alpha<K$ exists (see Theorem~\ref{theo_symm_and_linear_2}). 
As an easy consequence of the proof, we also get that exact PNE are guaranteed to exist if $n\leq 4$ or $m\leq 2$ (see Corollary~\ref{cor_2}). 

\subsection{An Algorithm for Computing $\alpha$-Approximate Equilibria}
For computing an $\alpha$-approximate PNE, we use the following basic approach. 
Starting with an empty game, we add the players one after another to the game, where a newly added player is always placed on a best response. If the addition of a new player makes some of the earlier added players ``unhappy'', meaning that they now have an $\alpha$-improving move, we let unhappy players deviate one after another to a best response until all players are happy again. Only then we add the next player, etc. 

Note that this approach has been used before to compute \emph{exact} PNE, for example for player-specific costs or weighted congestion games on matroids~\citep{Milchtaich96,Ackermann09}. However, to the best of our knowledge, it has not been utilized in the context of approximate PNE. We believe that this technique will be useful for showing existence and computing approximate equilibria beyond the class of games that we analyze here. 

For the formal description of our algorithm see Algorithm~\ref{alg:approxnash}. We assume that the resource set $R=\{r_1,\ldots,r_m\}$ is ordered such that $a_{1}\leq \cdots \leq a_{m}$ (where $a_j:=a_{r_j}$).
Since we need to consider strategy profiles for subsets of the players, let us extend the notion of a strategy profile to the set of vectors $x\in (R\cup\{0\})^n$, where $x_i=0$ means that player~$i$ has not yet been added to the game. Given such a vector $x\neq 0$, the cost $c_{x_i}(x)$ incurred to player~$i$ with $x_i\neq 0$ is defined as the cost which is experienced by her in the game where only the players~$j$ with $x_j\neq 0$ are present. 
Now define, for a strategy profile $x\in (R\cup\{0\})^n$, the following set of players~$W_{\alpha}(x)$ who are ``unhappy'' with their strategy, meaning that they have an $\alpha$-improving move: 
\begin{align*}
W_{\alpha}(x):=\{&i \in N: x_i\neq 0 \text{ and }\
c_{x_i}(x) > \alpha \cdot c_r(x_{-i},r) \\ &\text{for a resource } r \in R\}.
\end{align*}
In each iteration of the for-loop in line~\ref{for-loop} of Algorithm~\ref{alg:approxnash}, a new player~$k$ is added to the game. We place~$k$ on a best response (with respect to the strategies of the players $\{1,\ldots,k-1\}$ already added to the game). If there is more than one best response, we choose the one with smallest index, see lines~\ref{newplayer_br1} and~\ref{newplayer_br2}. 
After having added player~$k$, it may be the case that some players are not happy with their strategy, that is, $W_{\alpha}(x)\neq \emptyset$, where $x$ denotes the current strategy profile. 
In the while-loop starting in line~\ref{while}, we iteratively choose a player who is ``maximally unhappy'', meaning that she experiences maximum cost among all unhappy players, and let her deviate to a best response, until all players in~$\{1,\ldots,k\}$ are happy. If there is more than one resource with maximum cost among the resources used by unhappy players, we choose a player on the maximum cost resource with largest index, and if there is more than one best response, we choose the one with smallest index (see lines~\ref{while1}, \ref{while2} and~\ref{while3}). 
After this, we return to line~\ref{for-loop} where the next player~$k+1$ is added to the game. 

It is clear that if Algorithm~\ref{alg:approxnash} terminates, the computed strategy profile is an $\alpha$-PNE. In the next subsection, we show that Algorithm~\ref{alg:approxnash} terminates for $\alpha=K$. 
Note that the special choices made by the algorithm in case of non-unique best responses or non-unique most expensive resources, together with the fact that resources are ordered such that $a_{1} \leq a_{2} \leq \cdots \leq a_{m}$, ensure that the loads are always decreasing along the resources, that is, $\ell_{r_1}(x)  \geq \ell_{r_2}(x)  \geq \cdots \geq \ell_{r_m}(x)$ always holds during the algorithm. 

Appendix~\ref{subsec:example} contains an illustrating example for the application of Algorithm~\ref{alg:approxnash}.

\subsection{Termination of Algorithm~\ref{alg:approxnash} for $\alpha=K$}
In this subsection, we show that Algorithm~\ref{alg:approxnash} terminates for $\alpha=K$, and thus computes a $K$-approximate PNE. We also show that the running time of Algorithm~\ref{alg:approxnash} can be bounded by $\min\{O(n^2m),O(nm^2)\}$.

We need to prove that the while-loop terminates in each iteration~$k\in \{1,\ldots,n\}$ of the for-loop. 
For $k=1$, the while-loop obviously terminates since we only have one player, and she is placed on the best response $r_1$. Therefore, $W_{K}(x)=W_{K}(r_1,0,\ldots,0)=\emptyset$ and the while-loop terminates without changing anything. 
For $k=2$, the second player is placed on $r_1$ if $2a_1 + B \leq a_2 + B/2$, and is placed on $r_2$ otherwise. It is easy to see that in both cases, both players are happy and the while-loop immediately terminates. 

For iteration~$k\geq 3$, it may be the case that some players change their strategy during the while-loop. 
We show inductively that the while-loop also terminates in iteration~$k\geq 3$. Note that since the while-loop terminated in iteration~$k-1$, the players in $\{1,\ldots,k-1\}$ were happy before the next player~$k$ was added to the game. That is, with respect to the profile $x$ directly before player~$k$ is added, no player $i \in \{1, \ldots, k-1\}$ has a $K$-improving deviation, i.e., an alternative resource~$r$ with $c_{x_i}(x)>K\cdot c_r(x_{-i},r)$. 
But it may be the case that some players are unhappy after the addition of player~$k$. Additionally, the deviation of one of these players may cause further players to be unhappy. For the termination of the while-loop, we need to show that after finitely many deviations, all players in~$\{1,\ldots,k\}$ are happy again. We thus need to keep track of the set of unhappy players during the course of the while-loop. 
To this end, we derive necessary properties for players who either become unhappy due to the addition of player~$k$, or due to the subsequent deviation of some other player during the while-loop (see Lemma~\ref{lemma_cond1} and Lemma~\ref{lemma_cond2} in the appendix). 
The derived properties are mostly in terms of loads, e.g., conditions on the load of an unhappy player~$i$'s current strategy $x_i$, and on the load of a corresponding $K$-improving deviation~$r$. 
Using these two lemmas, as well as some further structural insights (see Appendix~\ref{subsec:lemmas}), we then proceed by a careful case distinction regarding the sequence of deviating players in iteration~$k$, and show that this sequence terminates in all cases. 
Let us now briefly sketch the mentioned case distinction (the complete proof of Theorem~\ref{theo_symm_and_linear_1} can be found in Appendix~\ref{appendix_maintheo}).  

Let $x$ and $x'$ denote the profiles directly before and after the new player~$k$ is added. 
If all players are happy with their strategy in $x'$, the statement follows; thus assume that player~$i$ changes from $x_i'=x_i$ to $r$ in the while-loop. 
Using Lemma~\ref{lemma_cond1}, we know that there are three possible cases regarding the load of player~$k$'s current strategy $x_k'$. Namely, the load $\ell_{x_k'}(x)$ of $x_k'$ (\emph{before} player~$k$ is added) needs to be in $\{M-2,M-1,M\}$, where $M:=M(x)$ denotes the maximum load with respect to $x$. 
We then analyze all three cases. 
As it turns out, the cases $\ell_{x_k'}(x)=M-2$ and $\ell_{x_k'}(x)=M-1$ are very simple, whereas $\ell_{x_k'}(x)=M$ is more complicated and requires further subcases. However, by repeated use of the lemmas contained in Appendix~\ref{subsec:lemmas}, in particular Lemma~\ref{lemma_cond2}, we can show termination of the while-loop also for this case.  

\begin{figure}[h]
\centering
\scalebox{0.8}{
\begin{tikzpicture}[decoration=brace]
\draw  (1, 0) rectangle (1.5, 1);
\draw[blue,dashed]  (1, 1) rectangle (1.5, 1.5);
\draw[red] (1.05, 1.05) rectangle (1.45, 1.45);
\node[below] at (1.25, -0.2) {$r_1$};
\draw  (3, 0) rectangle (3.5, 0.5);
\draw[red,dashed]  (3, 0.5) rectangle (3.5, 1);
\draw[green!50!black] (3.05, 0.55) rectangle (3.45, 0.95);
\node[below] at (3.25, -0.2) {$r_2$};
\draw[green!50!black,dashed]  (5, 0) rectangle (5.5, 0.5);
\draw[blue] (5.05, 0.05) rectangle (5.45, 0.45);
\node[below] at (5.25, -0.2) {$r_3$};
\path[->, >=stealth, blue] 
(1.25,1.5) edge[bend left=50]
node[above]{$2.$} (5.25,0.55);
\path[->, >=stealth, green!50!black] 
(5.25,0.5) edge[bend right=20]
node[above]{$1.$} (3.45,0.75);
\path[->, >=stealth, red]
(3.23,1) edge[bend right=20]
node[above]{$3.$} (1.45,1.25);
\end{tikzpicture}
}
\caption{A cycling sequence of deviations which might occur during the while-loop of Algorithm~\ref{alg:approxnash}: First, some player moves from $r_3$ to $r_2$, then another player changes her strategy from $r_1$ to $r_3$, and finally, a player using $r_2$ moves to~$r_1$.}
\label{fig_forK}
\end{figure}
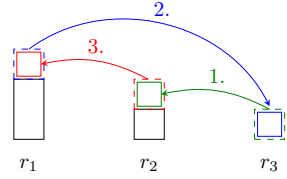

At the end of this proof sketch, we want to briefly indicate the role of the constant $K$. To this end, consider Figure~\ref{fig_forK} which shows a simplified version of a subcase occurring in the proof. In particular, the displayed sequence of deviations might occur during the while-loop of the algorithm, and this implies that the following three inequalities need to hold: 
\begin{align*}
a_3&>\alpha \cdot 2a_2 \\
3a_1+B&>\alpha \cdot  a_3\\
2a_2\cdot+B/2&>\alpha\cdot (3a_1+B) 
\end{align*}
From this, one can derive $(1 + \alpha^2/2 - \alpha^3)B > (\alpha^3 - 1) 3a_1$, which yields a contradiction for $\alpha = K$ and any $a_1 > 0$, as the left-hand side is equal to $0$, whereas the right-hand side is non-negative. Thus, cycles of this form cannot occur during the algorithm.

For a complete proof of Theorem~\ref{theo_symm_and_linear_1}, see Appendix~\ref{appendix_maintheo}.

\begin{theorem}\label{theo_symm_and_linear_1}
For $\alpha=K$, where $K=1/6 \cdot (1+\sqrt[3]{109-6\sqrt{330}}+\sqrt[3]{109+6\sqrt{330}})\approx 1.1974$ is the unique solution of the equation $-x^3+x^2/2+1=0$, Algorithm~\ref{alg:approxnash} computes a $K$-approximate PNE. 
\end{theorem}

The proof of Theorem~\ref{theo_symm_and_linear_1} also yields the following upper bound on the running time of Algorithm~\ref{alg:approxnash}.
\begin{corollary}
For $\alpha=K$, the running time of Algorithm~\ref{alg:approxnash} can be bounded by $\min\{O(n^2m),O(nm^2)\}$.
\end{corollary}

\begin{proof}
First note that each iteration of the while-loop can be implemented in $O(m)$ (note that although the cost of deviating to a resource $r$ is in general player-specific, since it depends on the load of the deviating player's current resource, it can in fact only be different for two players if one of these players is using a resource with maximum load, and the other not, cf. Lemma~\ref{lemma_deviation}). 
Furthermore, in the $k$th iteration of the for-loop, we can bound the number of iterations of the while-loop either by $O(m)$, or alternatively by $2k$ since no player moves more than twice, see the proof of Theorem~\ref{theo_symm_and_linear_1} in Appendix~\ref{appendix_maintheo}. 
Since there are $n$ iterations of the for-loop, we get $\min\{O(n^2),O(nm)\}$ as an upper bound for the total number of iterations of the while-loop and this yields the given bound on the total running time. 
\end{proof}

\begin{table*}[t]%
\begin{center}
\begin{tabular}{cp{6cm}c}
\toprule
load profile & $\alpha$-improving deviation (of some player \newline using $r$ to $r'$, notation $r\rightarrow r'$) & conditions on $\alpha$ \\ \midrule
$(5,0,0)$ &  $r_1 \rightarrow r_2$ & $5a_1+B>\alpha a_2 \ \Leftrightarrow \ \alpha<1/a_2=\frac{2}{K-1/2}\approx 2.86$ \\
$(4,1,0)$ & $r_1 \rightarrow r_2$ & $4a_1+B>\alpha 2a_2 \ \Leftrightarrow \ \alpha<1/(2a_2)=\frac{1}{K-1/2}\approx 1.43$ \\
$(3,2,0)$ & $r_1 \rightarrow r_3$ & $3a_1+B>\alpha a_3 \ \Leftrightarrow \ \alpha<1/a_3=K$ \\
$(3,1,1)$ & $r_3 \rightarrow r_2$ & $a_3>\alpha 2a_2 \ \Leftrightarrow \ \alpha<\frac{1}{K(K-1/2)}=K$ \\
$(2,2,1)$ & $r_2 \rightarrow r_1$ & $2a_2+B/2>\alpha (3a_1+B) \ \Leftrightarrow \ \alpha<2a_2+1/2=K$ \\\bottomrule
\end{tabular}
\caption{$\alpha$-improving deviations for the candidate profiles.}
\label{table_no_alpha_approx}
\end{center}
\end{table*}

The proof of Theorem~\ref{theo_symm_and_linear_1} also reveals that $n\geq 5$ and \mbox{$m\geq 3$} need to hold for any instance which has no exact PNE. 
This follows from the fact that the case displayed in Figure~\ref{fig_forK} essentially is the only situation where the while-loop might not terminate for $\alpha=1$, 
and this case requires at least five players and at least three resources (for a complete proof, see Appendix~\ref{appendix_maintheo}).
Thus we get the following corollary.

\begin{corollary}\label{cor_2}
Exact PNE are guaranteed to exist if $n\leq 4$ or $m\leq 2$. 
\end{corollary}  

\subsection{Tightness of $\alpha=K$}
\label{sec:tightness}
In this subsection, we provide an instance where no $\alpha$-approximate PNE with $\alpha<K$ exists, showing that $\alpha=K$ is the smallest possible value such that the existence of an $\alpha$-approximate equilibrium can be guaranteed.
\begin{theorem}\label{theo_symm_and_linear_2}
There exists an instance with three resources and five players 
such that there is no $\alpha$-approximate PNE for any $\alpha < K$, where $K \approx 1.1974$ is as specified in \eqref{def:K}. 
\end{theorem}
\begin{proof}
Consider the instance with $m=3$ resources, $n=5$ players, budget $B=1$, and resource cost coefficients 
\begin{align*}
a_1&:=a_{r_1}=0, \ a_2:=a_{r_2}=K/2-1/4 \approx 0.3487, \text{ and }\\ a_3&:=a_{r_3}=1/K \approx 0.8351.
\end{align*}
We proceed to show that there is no $\alpha$-approximate PNE for $\alpha<K$. 
To this end, note that it suffices to show that there is no $\alpha$-approximate PNE among the five load profiles $(5,0,0), (4,1,0), (3,2,0), (3,1,1)$, $(2,2,1)$ (since $a_1<a_2<a_3$; if there exists an $\alpha$-approximate PNE, there is also one with corresponding load profile among the five listed load profiles, see Lemma~\ref{lemma_decreasingload}). 
Let $\alpha<K$. We show that for any of the five load profiles, there exists an $\alpha$-improving deviation, showing the claim. To this end, consider Table~\ref{table_no_alpha_approx}, where we provide a deviation for each candidate profile which is $\alpha$-improving if the given conditions on $\alpha$ are satisfied. It is easy to check that these conditions are indeed fulfilled for $\alpha<K$.
\end{proof}

\section{Computing Optimal Approximate Equilibria}\label{sec:optimal}
In the last section, we showed that $\alpha=K$ is the smallest possible value for $\alpha$ such that the existence of an $\alpha$-approximate PNE can be \emph{guaranteed} for \emph{any} instance of a multi-leader congestion game with an adversary. However, there are clearly instances where $\alpha$-PNE with $\alpha<K$ exist (in particular, all instances exhibiting an exact PNE). We show in this section how to compute efficiently a best approximate PNE \emph{for a given instance}, that is, with smallest possible~$\alpha$ such that an $\alpha$-PNE exists for \emph{this} instance. 

To this end, consider a multi-leader congestion game with an adversary with resource set $R=[m]$ and $a_1\leq \cdots \leq a_m$. 
We can restrict our attention to strategy profiles~$x$ with decreasing loads, that is, with $\ell_1(x)\geq \cdots \geq \ell_m(x)$, since if an $\alpha$-PNE exists, there also exists one with decreasing loads (see Lemma~\ref{lemma_decreasingload} and note that we can assume $\alpha\in [1,K]$ since we want to find the smallest possible $\alpha$). 
Thus let $x$ be a strategy profile with $\ell_1(x)\geq \cdots \geq \ell_m(x)$. 
Note that, clearly, $M(x)=\max\{\ell_r(x): r\in R\}\geq \lceil \frac{n}{m} \rceil$ holds. Furthermore, if $M(x)m=n$ (which is equivalent to $\ell_r(x)=M(x)$ for all $r\in R$), $x$ is an $\alpha$-PNE if and only if $a_m\cdot M(x)+B/m\leq \alpha(a_1\cdot(M(x)+1)+B)$ holds. 
Thus we can assume in the following that there are resources with load $<M(x)$. 
We denote by $k=k(x)<m$ the largest resource having maximum load $M=M(x)$. Similarly, $k'=k'(x)$ denotes the smallest resource with load strictly smaller than $M-1$, and $k''=k''(x)$ denotes the smallest resource with load strictly smaller than $M-2$. 
In other words, $\ell_r(x)=M$ for all $r\in \{1,\ldots,k\}$, $\ell_r(x) = M-1$ for all $r\in \{k+1,\ldots,k'-1\}$, $\ell_r(x) = M-2$ for all $r\in \{k',\ldots,k''-1\}$, and $\ell_r(x)\leq M-3$ for all $r\in \{k'',\ldots,m\}$, see Figure~\ref{fig_5} for illustration. Note that $k'=k+1$ or $k'=k''$ are possible, in which case there are no resources with load $M-1$ or $M-2$, respectively.
\begin{figure*}[t]
\centering\scalebox{0.85}{
\begin{tikzpicture}[decoration=brace]
\node[below] at (-0.25, -0.2) {$R=[m]$:};
\draw  (1, 0) rectangle (1.5, 3.5);
\draw  (1, 3.5) rectangle (1.5, 4);
\node[below] at (1.25, -0.2) {$1$};
\node[below] at (2.25, 2) {$\ldots$};
\node[below] at (2.25, -0.42) {$\ldots$};
\draw  (3, 0) rectangle (3.5, 3.5); 
\draw  (3, 3.5) rectangle (3.5, 4);
\node[below][blue] at (3.25, -0.2) {$k$};
 \draw[decorate,  decoration={brace, mirror}]  (1,-1) -- node[below=0.6ex]{load $M$}  (3.5,-1);
 \draw  (4, 0) rectangle (4.5, 3);  
\draw  (4, 3) rectangle (4.5, 3.5);
\node[below] at (4.25, -0.2) {$k+1$};
\node[below] at (5.25, 1.75) {$\ldots$};
\node[below] at (5.25, -0.42) {$\ldots$};
 \draw  (6, 0) rectangle (6.5, 3); 
\draw  (6, 3) rectangle (6.5, 3.5);
\node[below] at (6.25, -0.2) {$k'-1$};
 \draw[decorate,  decoration={brace, mirror}]  (4,-1) -- node[below=0.6ex]{load $M-1$}  (6.5,-1);
\draw  (7, 0) rectangle (7.5, 2.5); 
\draw  (7, 2.5) rectangle (7.5, 3);
\node[below][green!50!black] at (7.25, -0.2) {$k'$};
\node[below] at (8.25, 1.5) {$\ldots$};
\node[below] at (8.25, -0.42) {$\ldots$};
\draw  (9, 0) rectangle (9.5, 2.5); 
\draw  (9, 2.5) rectangle (9.5, 3);
\node[below] at (9.25, -0.2) {$k''-1$};
\draw[decorate,  decoration={brace, mirror}]  (7,-1) -- node[below=0.6ex]{load $M-2$}  (9.5,-1);
\draw  (10, 0) rectangle (10.5, 2.5); 
\node[below][magenta] at (10.25, -0.2) {$k''$};
\node[below] at (11.25, 1.25) {$\ddots$};
\node[below] at (11.25, -0.42) {$\ldots$};
\draw  (12, 0) rectangle (12.5, 1); 
\node[below] at (12.25, -0.2) {$m$};
\draw[decorate,  decoration={brace, mirror}]  (10,-1) -- node[below=0.6ex]{load $\leq M-3$}  (12.5,-1);
\end{tikzpicture}
}
\caption{Illustration for the definition of \textcolor{blue}{$k$}, \textcolor{green!50!black}{$k'$} and \textcolor{magenta}{$k''$}.}
\label{fig_5}
\end{figure*}
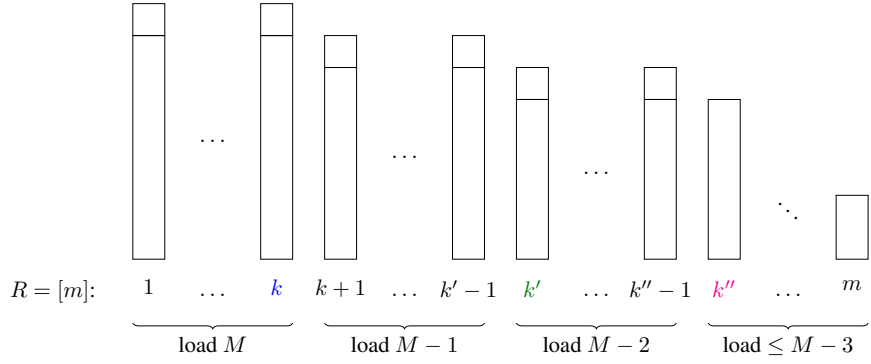

We now define the following values $\bar{c}_M=\bar{c}_M(x)$ and $\bar{c}_{< M}=\bar{c}_{< M}(x)$, which essentially describe the cost of a best alternative resource for a player using a resource with load $M$, and load smaller than $M$, respectively:
\[
\bar{c}_M=\begin{cases}  
   \min  \{a_{k+1}\cdot M+B/k \text{ if }k'\geq k+2, &\\
 \phantom{\min \{} a_{k'}\cdot (M-1)+B/k' \text{ if } k'<k'',&\\
\phantom{\min \{} a_r\cdot(\ell_r(x)+1)\ \forall \ r\geq k''\}, 
   &\textrm{if $k=1$,}\\

\\
\min \{a_{1}\cdot (M+1)+B, & \\
 \phantom{\min \{}a_{k+1}\cdot M+B/k \text{ if }k'\geq k+2, &\\
 \phantom{\min \{} a_{k'}\cdot (M-1) \text{ if } k'<k'', &\\
 \phantom{\min \{}a_r\cdot(\ell_r(x)+1) \ \forall \ r\geq k''\},  
&\textrm{if $k\geq 2$,}\\
\end{cases}
\]
and 
\begin{align*}
\bar{c}_{< M}=\min\{&a_{1}\cdot (M+1)+B, \\
&a_{k+1}\cdot M+B/(k+1) \text{ if }k'\geq k+2, \\
&a_{k'}\cdot (M-1) \text{ if } k'<k'', \\
&a_r\cdot(\ell_r(x)+1) \ \forall \ r\geq k''\}.
\end{align*}

Using these definitions, $x$ is an $\alpha$-approximate PNE iff
\begin{align}
    a_r\cdot\ell_r(x)\leq \alpha\cdot \bar{c}_{< M} \text{ for all } r>k, \text{ and }\label{eq_nash_1} \\ a_k\cdot M+B/k\leq\alpha\cdot  \bar{c}_M. \label{eq_nash_2}
\end{align}
Regarding this, note that there are some cases in which $\bar{c}_M$ or $\bar{c}_{< M}$ do not denote the cost of a best alternative for some players, but only provide a lower bound. 
This happens in two cases. 
Firstly, if $k\geq 2$ and the minimum in the definition of $\bar{c}_M$ is uniquely attained at $a_{1}\cdot (M+1)+B$, then $\bar{c}_M$ is in general not the cost of a best alternative for players currently using resource $r=1$. However, it is then also clear that these players are satisfied with their strategies anyways. 
Similarly, if the minimum in the definition of $\bar{c}_{<M}$ is uniquely attained for some resource $r>k$, then $\bar{c}_{<M}$ does not denote the cost of a best alternative for the players using $r$. But it is again clear that these players do not want to deviate. 

Also note that the minimum properties of  $\bar{c}_M$ and $\bar{c}_{<M}$, as well as the Nash condition~\eqref{eq_nash_1}, yield upper and lower bounds for the loads $\ell_r(x)$ of all resources $r\geq k''$ (for $r<k''$, the load $\ell_r(x)$ is uniquely determined by the definitions of $M,k,k'$ and $k''$).  
By using this, as well as the fact that for given values of $M,k,k'$ and $k''$, there are only polynomially many possible values for $\bar{c}_M$ and $\bar{c}_{<M}$, we can show the following result (for a complete proof, see Appendix~\ref{appendix_sec4}):

\begin{lemma}\label{lemma_existence}
Given $M\in \{\lceil \frac{n}{m} \rceil,\ldots,n\}$, $k\in \{1,\ldots,m-1\}$, $k'\in \{k+1,\ldots,m\}$ and $k''\in \{k',\ldots,m\}$, as well as $\alpha \in [1,K]$, 
we can decide efficiently if there exists an $\alpha$-approximate PNE $x$ such that $M(x)=M$, $k(x)=k$, $k'(x)=k'$ and $k''(x)=k''$ holds. 

In case of existence, we can furthermore compute a corresponding load vector efficiently.
\end{lemma}

Using Lemma~\ref{lemma_existence}, we get the main result of this section:
\begin{theorem}
We can efficiently compute the smallest possible $\alpha$ such that an $\alpha$-approximate PNE exists, as well as a corresponding load vector. 
\end{theorem}
\begin{proof}
First note that if $x$ is an $\alpha$-PNE corresponding to the best possible $\alpha$, there need to be resources $r,r'\in R$ (as well as a player~$i$ with $x_i=r$) such that $a_r\ell_r(x)+\kappa^*_r(x)=\alpha\cdot(a_{r'}\ell_{r'}(r',x_{-i})+\kappa^*_{r'}(r',x_{-i}))$ (otherwise $\alpha$ cannot be smallest possible), or, equivalently, 
\[\alpha=\frac{a_r\ell_r(x)+\kappa^*_r(x)}{a_{r'}\ell_{r'}(r',x_{-i})+\kappa^*_{r'}(r',x_{-i})}.\] 
Since there are only $O(nm^2)$ many possible values for $a_r\ell_r(x)+\kappa^*_r(x)$, as well as for $a_{r'}\ell_{r'}(r',x_{-i})+\kappa^*_{r'}(r',x_{-i})$, there are only $O(n^2m^4)$ many possible values for $\alpha$. 

The result then follows by applying the underlying procedure of Lemma~\ref{lemma_existence} for all $M\in \{\lceil \frac{n}{m} \rceil,\ldots,n\}$, $k\in \{1,\ldots,m-1\}$, $k'\in \{k+1,\ldots,m\}$ and $k''\in \{k',\ldots,m\}$, as well as the $O(n^2m^4)$ many possible values for $\alpha$ (thus in total $O(n^3m^7)$ many times). 
\end{proof}

\section{Conclusion}
We introduced a multi-leader congestion game with an adversary which is motivated by security applications with congestion effects. 
Since PNE do not exist in general, we studied approximate equilibria. Our first main result shows that a $K$-approximate PNE always exists, where $K\approx1.1974$ is the unique solution of a cubic polynomial equation. To this end, we presented an efficient algorithm which computes a $K$-approximate PNE. Furthermore, we showed that the factor $K$ is tight by providing an instance where no $\alpha$-approximate PNE with $\alpha<K$ exists. 
However, for a \emph{specific} instance there might be a better $\alpha$-approximate PNE, i.e., with $\alpha<K$. We presented an efficient procedure that computes a \emph{best} approximate PNE of a given instance.

Our work also suggests several interesting directions for further research regarding multi-leader congestion games with an adversary. 
For example, one could analyze whether the results from Section~\ref{sec:unrestricted} continue to hold if one allows more general strategy spaces in the leaders' congestion game. A first natural generalization in this regard would be to consider asymmetric strategies, and/or bases of matroids. 
It would furthermore be interesting to analyze the quality of approximate PNE. 
For example, one could measure the social cost of a strategy profile by the total cost of all players, and then compare a (best or worse) approximate PNE to a social optimum. 

\clearpage

\paragraph{Acknowledgements}
We thank the anonymous referees for their comments that helped to improve the presentation of the paper. This work was supported by Deutsche Forschungsgemeinschaft (DFG -- German Research Foundation) under grants HA 8041/4-1, MA 8439/1-1, and under Germany's Excellence Strategy -- The Berlin Mathematics Research Center MATH+ (EXC-2046/1, project ID: 390685689).

\bibliography{master-bib.bib}

\begin{thebibliography}{27}
\providecommand{\natexlab}[1]{#1}

\bibitem[{Ackermann, R\"{o}glin, and V\"{o}cking(2009)}]{Ackermann09}
Ackermann, H.; R\"{o}glin, H.; and V\"{o}cking, B. 2009.
\newblock Pure {N}ash equilibria in player-specific and weighted congestion
  games.
\newblock \emph{Theoret. Comput. Sci.}, 410(17): 1552--1563.

\bibitem[{Babaioff, Kleinberg, and Papadimitriou(2009)}]{BabaioffKP09}
Babaioff, M.; Kleinberg, R.; and Papadimitriou, C.~H. 2009.
\newblock Congestion games with malicious players.
\newblock \emph{Games Econ. Behav.}, 67(1): 22--35.

\bibitem[{Beckmann, McGuire, and Winsten(1956)}]{Beckmann56}
Beckmann, M.; McGuire, C.; and Winsten, C. 1956.
\newblock \emph{Studies in the Economics and Transportation}.
\newblock New Haven, CT, USA: Yale University Press.

\bibitem[{Bil{\`{o}}, Moscardelli, and Vinci(2018)}]{BiloMV18}
Bil{\`{o}}, V.; Moscardelli, L.; and Vinci, C. 2018.
\newblock Uniform Mixed Equilibria in Network Congestion Games with Link
  Failures.
\newblock In Chatzigiannakis, I.; Kaklamanis, C.; Marx, D.; and Sannella, D.,
  eds., \emph{45th International Colloquium on Automata, Languages, and
  Programming, {ICALP} 2018, July 9-13, 2018, Prague, Czech Republic}, volume
  107 of \emph{LIPIcs}, 146:1--146:14. Schloss Dagstuhl - Leibniz-Zentrum
  f{\"{u}}r Informatik.

\bibitem[{Caragiannis et~al.(2011)Caragiannis, Fanelli, Gravin, and
  Skopalik}]{CaragiannisFGS11}
Caragiannis, I.; Fanelli, A.; Gravin, N.; and Skopalik, A. 2011.
\newblock Efficient Computation of Approximate Pure Nash Equilibria in
  Congestion Games.
\newblock In Ostrovsky, R., ed., \emph{{IEEE} 52nd Annual Symposium on
  Foundations of Computer Science, {FOCS} 2011, Palm Springs, CA, USA, October
  22-25, 2011}, 532--541. {IEEE} Computer Society.

\bibitem[{Caragiannis et~al.(2015)Caragiannis, Fanelli, Gravin, and
  Skopalik}]{CaragiannisFGS15}
Caragiannis, I.; Fanelli, A.; Gravin, N.; and Skopalik, A. 2015.
\newblock Approximate Pure Nash Equilibria in Weighted Congestion Games:
  Existence, Efficient Computation, and Structure.
\newblock \emph{{ACM} Trans. Economics and Comput.}, 3(1): 2:1--2:32.

\bibitem[{Castiglioni et~al.(2019)Castiglioni, Marchesi, Gatti, and
  Coniglio}]{CastiglioniMGC19}
Castiglioni, M.; Marchesi, A.; Gatti, N.; and Coniglio, S. 2019.
\newblock Leadership in singleton congestion games: What is hard and what is
  easy.
\newblock \emph{Artif. Intell.}, 277: 103177.

\bibitem[{Correa et~al.(2018)Correa, Guzm{\'{a}}n, Lianeas, Nikolova, and
  Schr{\"{o}}der}]{CorreaGLNS18}
Correa, J.~R.; Guzm{\'{a}}n, C.; Lianeas, T.; Nikolova, E.; and Schr{\"{o}}der,
  M. 2018.
\newblock Network Pricing: How to Induce Optimal Flows Under Strategic Link
  Operators.
\newblock In \emph{Proc. 19th ACM Conf. Electronic Commerce (EC)}, 375--392.

\bibitem[{Correa et~al.(2017)Correa, Harks, Kreuzen, and Matuschke}]{CHKM17}
Correa, J.~R.; Harks, T.; Kreuzen, V. J.~C.; and Matuschke, J. 2017.
\newblock Fare Evasion in Transit Networks.
\newblock \emph{Oper. Res.}, 65(1): 165--183.

\bibitem[{Gairing, Harks, and Klimm(2017)}]{GairingHK17}
Gairing, M.; Harks, T.; and Klimm, M. 2017.
\newblock Complexity and Approximation of the Continuous Network Design
  Problem.
\newblock \emph{SIAM J. Optim.}, 27(3): 1554--1582.

\bibitem[{Gan, Elkind, and Wooldridge(2018)}]{GanEW18}
Gan, J.; Elkind, E.; and Wooldridge, M.~J. 2018.
\newblock Stackelberg Security Games with Multiple Uncoordinated Defenders.
\newblock In Andr{\'{e}}, E.; Koenig, S.; Dastani, M.; and Sukthankar, G.,
  eds., \emph{Proceedings of the 17th International Conference on Autonomous
  Agents and MultiAgent Systems, {AAMAS} 2018, Stockholm, Sweden}, 703--711.
  International Foundation for Autonomous Agents and Multiagent Systems
  Richland, SC, {USA} / {ACM}.

\bibitem[{Harks and Schedel(2019)}]{HarksSchedel19}
Harks, T.; and Schedel, A. 2019.
\newblock Capacity and Price Competition in Markets with Congestion Effects.
\newblock In \emph{Proc. 15th Internat. Conference on Web and Internet
  Econom.}, 341.

\bibitem[{Harks, Schr{\"{o}}der, and Vermeulen(2019)}]{HarksSV19}
Harks, T.; Schr{\"{o}}der, M.; and Vermeulen, D. 2019.
\newblock Toll caps in privatized road networks.
\newblock \emph{Eur. J. Oper. Res.}, 276(3): 947 -- 956.

\bibitem[{Johari, Weintraub, and {Van Roy}(2010)}]{JohariWR10}
Johari, R.; Weintraub, G.~Y.; and {Van Roy}, B. 2010.
\newblock Investment and Market Structure in Industries with Congestion.
\newblock \emph{Oper. Res.}, 58(5): 1303--1317.

\bibitem[{Kiekintveld et~al.(2009)Kiekintveld, Jain, Tsai, Pita,
  Ord\'{o}\~{n}ez, and Tambe}]{Kiekintveld09}
Kiekintveld, C.; Jain, M.; Tsai, J.; Pita, J.; Ord\'{o}\~{n}ez, F.; and Tambe,
  M. 2009.
\newblock Computing Optimal Randomized Resource Allocations for Massive
  Security Games.
\newblock In \emph{Proceedings of The 8th International Conference on
  Autonomous Agents and Multiagent Systems - Volume 1}, AAMAS '09, 689–696.
  Richland, SC: International Foundation for Autonomous Agents and Multiagent
  Systems.

\bibitem[{{Kulkarni} and {Shanbhag}(2015)}]{Kulkarni15}
{Kulkarni}, A.~A.; and {Shanbhag}, U.~V. 2015.
\newblock An Existence Result for Hierarchical {S}tackelberg v/s {S}tackelberg
  Games.
\newblock \emph{IEEE Transactions on Automatic Control}, 60(12): 3379--3384.

\bibitem[{Labb{\'e}, Marcotte, and Savard(1998)}]{Labbe98}
Labb{\'e}, M.; Marcotte, P.; and Savard, G. 1998.
\newblock A Bilevel Model of Taxation and Its Application to Optimal Highway
  Pricing.
\newblock \emph{Management Science}, 44(12): 1608--1622.

\bibitem[{Leyffer and Munson(2010)}]{LeyfferM10}
Leyffer, S.; and Munson, T. 2010.
\newblock Solving multi-leader-common-follower games.
\newblock \emph{Optim. Methods Softw.}, 25(4): 601--623.

\bibitem[{Li et~al.(2017)Li, Jia, Tan, Wang, Han, and Lau}]{Yupeng17}
Li, Y.; Jia, Y.; Tan, H.; Wang, R.; Han, Z.; and Lau, F. C.~M. 2017.
\newblock Congestion Game With Agent and Resource Failures.
\newblock \emph{IEEE Journal on Selected Areas in Communications}, 35(3):
  764--778.

\bibitem[{Liu, Chen, and Huang(2011)}]{Liu2011}
Liu, T.-L.; Chen, J.; and Huang, H.-J. 2011.
\newblock Existence and efficiency of oligopoly equilibrium under toll and
  capacity competition.
\newblock \emph{Transportation Research Part E: Logistics and Transportation
  Review}, 47(6): 908 -- 919.

\bibitem[{Marchesi, Castiglioni, and Gatti(2019)}]{MarchesiC019}
Marchesi, A.; Castiglioni, M.; and Gatti, N. 2019.
\newblock Leadership in Congestion Games: Multiple User Classes and
  Non-Singleton Actions.
\newblock In Kraus, S., ed., \emph{Proc. 28th Internat. Joint Conf. Artif.
  Intell. (IJCAI)}, 485--491.

\bibitem[{Marcotte(1986)}]{Marcotte85mp}
Marcotte, P. 1986.
\newblock Network Design Problem with Congestion Effects: A Case of Bilevel
  Programming.
\newblock \emph{Math. Program., Ser. A}, 34: 142--162.

\bibitem[{Meir et~al.(2012)Meir, Tennenholtz, Bachrach, and Key}]{MeirTBK12}
Meir, R.; Tennenholtz, M.; Bachrach, Y.; and Key, P.~B. 2012.
\newblock Congestion Games with Agent Failures.
\newblock In Hoffmann, J.; and Selman, B., eds., \emph{Proceedings of the
  Twenty-Sixth {AAAI} Conference on Artificial Intelligence, July 22-26, 2012,
  Toronto, Ontario, Canada}. {AAAI} Press.

\bibitem[{Milchtaich(1996)}]{Milchtaich96}
Milchtaich, I. 1996.
\newblock Congestion Games with Player-Specific Payoff Functions.
\newblock \emph{Games Econom. Behav.}, 13(1): 111--124.

\bibitem[{Sinha et~al.(2018)Sinha, Fang, An, Kiekintveld, and
  Tambe}]{SinhaFAKT18}
Sinha, A.; Fang, F.; An, B.; Kiekintveld, C.; and Tambe, M. 2018.
\newblock Stackelberg Security Games: Looking Beyond a Decade of Success.
\newblock In Lang, J., ed., \emph{Proc. 27th Internat. Joint Conf. Artif.
  Intell. (IJCAI)}, 5494--5501.

\bibitem[{Wardrop(1952)}]{Wardrop52}
Wardrop, J. 1952.
\newblock Some theoretical aspects of road traffic research.
\newblock \emph{Proc. Inst. Civil Engineers}, 1(Part II): 325--378.

\bibitem[{Yang and Huang(2004)}]{Yang04}
Yang, H.; and Huang, H.-J. 2004.
\newblock The multi-class, multi-criteria traffic network equilibrium and
  systems optimum problem.
\newblock \emph{Transportation Res.}, 38(B): 1--15.

\end{thebibliography}

\clearpage

\appendix

\section{An Illustrating Example for Algorithm~\ref{alg:approxnash}}\label{subsec:example}
\begin{appexample}
Consider a game with seven players and five resources with $a_{1} = 1$, $a_2 = 4$, $a_3 = 4$, $a_4 = 10$, and $a_5 =10$. Let $B = 9$ and $\alpha = K \approx 1.1974$ as specified in \eqref{def:K}.
We now describe the steps performed by Algorithm~\ref{alg:approxnash}, see also Figure~\ref{fig_example}. 
Note that in the first six iterations of the for-loop, no players change during the while-loop. Only after player~$7$ is added, the players~$5$ and $6$ become unhappy. Since the players~$5$ and $6$ currently experience the same cost and player~$6$ uses the resource with a larger index, player~$6$ changes from $r_5$ to $r_2$ due to the tie-breaking rule of the algorithm. Note that the resources $r_2$ and $r_3$ both constitute a best response but $r_2$ has the smaller index. The deviation of player~$6$ causes the players on $r_1$ to be unhappy since they can improve their cost from $12$ to $10$ (more than a factor $K$) by deviating to $r_3$. Furthermore player~$5$ can improve her cost from $10$ to $8$. Since $12>10$ a player on $r_1$ changes due to the tie-breaking rule of the algorithm. The resulting strategy profile if one of these players, say player~$7$, deviates to $r_3$, is a $K$-approximate PNE since no player has a $K$-improving deviation anymore. (In fact, the profile even is a $12.5/12=25/24$-PNE.)  

Note that if some player now deviates from $r_2$ to $r_1$ (which is a best response), the resulting strategy profile is not a $K$-approximate PNE anymore---the situation is essentially as in Figure~\ref{fig_example}~(a). This also shows that, starting from a $K$-PNE, a sequence of best responses does not necessarily yield a $K$-PNE.
\begin{figure}[h]
\begin{center}
\subfloat[Situation after player~7 is added. In the while-loop, player~$6$ wants to deviate to $r_2$ since $10>K\cdot 8\approx 9.58$.]{
\begin{tikzpicture}[decoration=brace]
\draw  (1, 0) rectangle (1.5, 0.5);
\node at (1.25,0.25) {$1$};
\draw  (1, 0.5) rectangle (1.5, 1);
\node at (1.25,0.75) {$4$};
\draw  (1, 1) rectangle (1.5, 1.5);
\node at (1.25,1.25) {$7$};
\node[below] at (1.25, -0.2) {$r_1$};
\draw  (2, 0) rectangle (2.5, 0.5); 
\node at (2.25,0.25) {$2$};
\node[below] at (2.25, -0.2) {$r_2$};
\draw  (3, 0) rectangle (3.5, 0.5); 
\node at (3.25,0.25) {$3$};
\node[below] at (3.25, -0.2) {$r_3$};
\draw  (3, 0) rectangle (3.5, 0.5); 
\node at (4.25,0.25) {$5$};
\node[below] at (4.25, -0.2) {$r_4$};
\draw  (4, 0) rectangle (4.5, 0.5); 
\node at (5.25,0.25) {$6$};
\node[below] at (5.25, -0.2) {$r_5$};
\draw  (5, 0) rectangle (5.5, 0.5); 
\node at (0,0) {};
\node at (4,0) {};
\end{tikzpicture}
}
\hspace{1cm}
\subfloat[Situation after player~$6$ deviated to $r_2$. Player~$7$ now wants to deviate to $r_5$ since $12>K\cdot 10 \approx 11.97$.]{
\begin{tikzpicture}[decoration=brace]
\draw  (1, 0) rectangle (1.5, 0.5);
\node at (1.25,0.25) {$1$};
\draw  (1, 0.5) rectangle (1.5, 1);
\node at (1.25,0.75) {$4$};
\draw  (1, 1) rectangle (1.5, 1.5);
\node at (1.25,1.25) {$7$};
\node[below] at (1.25, -0.2) {$r_1$};
\draw  (2, 0) rectangle (2.5, 0.5); 
\draw  (2, 0.5) rectangle (2.5, 1);
\node at (2.25,0.25) {$2$};
\node at (2.25,0.75) {$6$};
\node[below] at (2.25, -0.2) {$r_2$};
\draw  (3, 0) rectangle (3.5, 0.5); 
\node at (3.25,0.25) {$3$};
\node[below] at (3.25, -0.2) {$r_3$};
\draw  (3, 0) rectangle (3.5, 0.5); 
\node at (4.25,0.25) {$5$};
\node[below] at (4.25, -0.2) {$r_4$};
\draw  (4, 0) rectangle (4.5, 0.5); 
\node[below] at (5.25, -0.2) {$r_5$};
\node at (0,0) {};
\node at (4,0) {};
\end{tikzpicture}
}
\hspace{1cm}
\subfloat[Situation after player~$7$ deviated to $r_5$. The while-loop terminates since no player has a $K$-improving deviation.]{
\begin{tikzpicture}[decoration=brace]
\draw  (1, 0) rectangle (1.5, 0.5);
\node at (1.25,0.25) {$1$};
\draw  (1, 0.5) rectangle (1.5, 1);
\node at (1.25,0.75) {$4$};
\node[below] at (1.25, -0.2) {$r_1$};
\draw  (2, 0) rectangle (2.5, 0.5); 
\draw  (2, 0.5) rectangle (2.5, 1);
\node at (2.25,0.25) {$2$};
\node at (2.25,0.75) {$6$};
\node[below] at (2.25, -0.2) {$r_2$};
\draw  (3, 0) rectangle (3.5, 0.5); 
\draw  (3, 0) rectangle (3.5, 0.5); 
\node at (3.25,0.25) {$3$};
\node[below] at (3.25, -0.2) {$r_3$};
\draw  (4, 0) rectangle (4.5, 0.5);
\node at (4.25,0.25) {$5$};
\node[below] at (4.25, -0.2) {$r_4$};
\draw  (5, 0) rectangle (5.5, 0.5);
\node at (5.25,0.25) {$7$};
\node[below] at (5.25, -0.2) {$r_5$};
\node at (0,0) {};
\node at (4,0) {};
\end{tikzpicture}
}
\end{center}
\caption{Example for Algorithm~\ref{alg:approxnash} with $n=7$ players, and five resources $r_1,\dots,r_5$. Player~$i\in [n]$ is represented by a square containing number~$i$.}%
\label{fig_example}%
\end{figure}
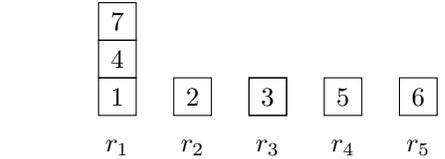
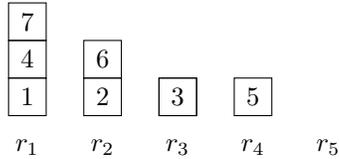
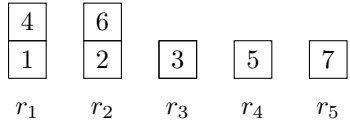
\end{appexample}

\section{Omitted Lemmas from Section~\ref{sec:unrestricted} }\label{subsec:lemmas}
This section contains various lemmas which are used to prove Theorem~\ref{theo_symm_and_linear_1}.
Before we state the Lemmas~\ref{lemma_cond1} and~\ref{lemma_cond2} which are the basis of the case distinction used in the proof of Theorem~\ref{theo_symm_and_linear_1}, we introduce the following notation, and provide some simple lemmas. 
\begin{notation}
Given a strategy profile~$x$, a player~$i$ with $x_i\neq 0$ and a resource~$r\neq x_i$, we write $\dev(i,r,x):=c_r(x_{-i},r)$ for the cost experienced by player~$i$ if she unilaterally deviates from her current strategy~$x_i$ to a different resource~$r$. We also call $\dev(i,r,x)$ the deviation cost of player~$i$ for $r$ (with respect to $x$). 
\end{notation}
\begin{applemma}\label{lemma_deviation}
Let $x$ be a strategy profile, $i$ and $j$ two players with strategies $x_i,x_j\in R$, and $r$ a resource different from $x_i$ and $x_j$. Furthermore assume that $\ell_{x_i}(x)\geq \ell_{x_j}(x)$, and define $M:=M(x)=\max\{\ell_{r'}(x): r'\in R\}$. 
Then: 
\begin{align*}
\dev(i,r,x) \neq \dev(j,r,x) \\ \Leftrightarrow \\  \text{ One of the following two cases holds:} \\
\ell_r(x)=M-1, \ell_{x_i}(x)=M,  \text{ and }\ell_{x_j}(x)<M, \ or \\ 
\ell_r(x)=M-2 \text{ and $x_i$ the only resource with load $M$ in $x$.}
\end{align*}
\end{applemma}
\begin{proof}
If a player~$i$ with $\ell_{x_i}(x)=M$ changes from her current strategy $x_i$ to $r \neq x_i$ with $\ell_r(x)=M-1$, she has to pay $\dev(i,r,x)=a_r (\ell_r(x)+1)+\frac{B}{p}$, where $p:=|M^{-1}(x)|$ denotes the number of resources with load $M$ in $x$. If a player~$j$ with $\ell_{x_j}(x) < M$ changes from her current strategy $x_j$ to $r \neq x_j$ with $\ell_r(x)=M-1$, she has to pay $\dev(j,r,x)=a_r (\ell_r(x)+1)+\frac{B}{p+1}$. Thus, $\dev(i,r,x) \ne \dev(j,r,x)$. If $x_i$ is the only resource with load $M$ in $x$ and player $i$ changes from her current strategy $x_i$ to $r \neq x_i$ with $\ell_r(x)=M-2$, she has to pay $\dev(i,r,x)=a_r (\ell_r(x)+1)+\frac{B}{\tilde{p}+1}$, where $\tilde{p}$ is the number of resources with load $M-1$ in $x$. If a player~$j$ changes from her current strategy $x_j \neq x_i$ to $r \neq x_j$ with $\ell_r(x)=M-2$, she has to pay $\dev(j,r,x)=a_r (\ell_r(x)+1)$. Thus, $\dev(i,r,x) \ne \dev(j,r,x)$.

Now we will show that if $\dev(i,r,x) \neq \dev(j,r,x)$, then either
$\ell_r(x)=M-1$, $\ell_{x_i}(x)=M$ and $\ell_{x_j}(x)<M$ or 
$\ell_r(x)=M-2$ and $x_i$ is the only resource with load $M$ in $x$. We consider a deviation to $r \neq x_i,x_j$ for player~$i$ and $j$ in $x$, respectively. It is clear that the congestion part of the deviation cost of player~$i$ or $j$ for $r$, which is $a_r (\ell_r(x)+1)$, is the same. Since we assume $\dev(i,r,x) \neq \dev(j,r,x)$, the adversary part, which is $\kappa_r^*(x_{-i},r)$ and $\kappa_r^*(x_{-j},r)$, must be different. In particular, the adversary parts cannot be $0$ for both. We write $M_{i,r}:=\max\{\ell_{r'}(x_{-i},r):r' \in R\}$ for the maximum load after player~$i$ deviated from her current strategy $x_i$ to $r$. We now distinguish between the possible values for $M_{i,r} \in \{M-1,M,M+1\}$. 

Case $M_{i,r}=M+1$: This case implies $\ell_r(x)=M$ and $\ell_r(x_{-i},r)=\ell_r(x_{-j},r)=M+1$. Thus, both players have to pay the whole adversary budget $B$ after changing to $r$ in $x$. Remark that the other resources have load at most $M$ after player~$i$ or $j$ deviates. This contradicts our assumption that the adversary parts must be different for the two players.

Case $M_{i,r}=M-1$: This case implies that $x_i$ is the only resource with load $M$ in $x$ and $\ell_r(x) \leq M-2$. If player~$j$ changes from her current strategy $x_j \neq x_i$ to $r$ the adversary part is $0$ (the resource $x_i$ has still load $M$). Since the adversary part should be different after the deviations to the resource $r$, we conclude that the adversary part of the deviation cost of player~$i$ must be positive. Therefore, $\ell_r(x)+1=M-1$ which is equivalent to $\ell_r(x)=M-2$.

Case $M_{i,r}=M$: This case implies $\ell_r(x) \leq M-1$ (if $\ell_r(x)=M$ we would get $M_{i,r}=M+1$). If $\ell_r(x) \leq M-2$ and player~$i$ changes to $r$ the adversary part of player $i$'s  cost is $0$. Note that $x_i$ can not be the only resource with load $M$ which would imply $M_{i,r}=M-1$. Since the adversary part should be different after the deviations to the resource $r$, we conclude that the adversary part of the deviation cost of player~$j$ must be positive. Therefore, $x_j$ must be the only resource with load $M$ in $x$ and $\ell_r(x)=M-2$. Thus, $\ell_{x_j}(x) > \ell_{x_i}(x)$ which contradicts the assumption that $\ell_{x_j}(x) \leq \ell_{x_i}(x)$. Hence the case $\ell_r(x) \leq M-2$ can not occur and we know $\ell_r(x) = M-1$. This implies $\ell_r(x_{-i},r)=\ell_r(x_{-j},r)=M$ and both adversary parts of the deviation costs of player~$i$ and $j$ for the resource $r$ must be positive. Since the adversary part must be different the number of resources with maximal load must be different after the respective deviations to $r$. Let $p:=|M^{-1}(x)|$ be the number of resources with load $M$ in $x$, and $p_i$ and $p_j$ the number of resources with load $M$ in $(x_{-i},r)$ and $(x_{-j},r)$, respectively. As mentioned above the adversary part must be different, that is, $p_i \ne p_j$. If $\ell_{x_i}(x)=M$ we get $p_i=p$ (note that $\ell_r(x) = M-1$). If, additionally, $\ell_{x_j}(x)=M$, we get $p_j=p=p_i$ which is a contradiction to our assumption. Since $\ell_{x_j}(x) \leq \ell_{x_i}(x)$ we can conclude $\ell_{x_j}(x) \leq M-1$. Finally, we exclude the case $\ell_{x_i}(x) \leq M-1$. Since we assume $\ell_{x_j}(x) \leq \ell_{x_i}(x)\le M-1$, we get $p_i=p+1=p_j$ (note that $\ell_r(x) = M-1$), which is a contradiction to our assumption.
\end{proof}

\begin{applemma}\label{lemma_load-1}
Let $x$ be a strategy profile occurring in Algorithm~\ref{alg:approxnash}, $i\in N$ a player with strategy $x_i\in R$, and $r\in R$ a resource with load $\ell_r(x)= \ell_{x_i}(x)-1$. 
Then player~$i$ does not want to deviate to~$r$.  
\end{applemma}
\begin{proof}
We show that a deviation to $r$ can only increase player~$i$'s cost. Clearly, the adversary part of the cost remains unchanged. Furthermore, the congestion part can only be larger, since $a_{x_i}\leq a_r$ holds: Recall that loads are always decreasing along the resources, that is, $\ell_{r_1}(x)  \geq \ell_{r_2}(x) \geq \cdots \geq \ell_{r_m}(x)$ holds, and furthermore $a_{r_1}\leq a_{r_2}\leq \cdots \leq a_{r_m}$. Altogether, the deviation is not beneficial for player~$i$. 
\end{proof}

\begin{applemma}\label{lemma_x_i_preference}
Let $x$ and $x'$ be two strategy profiles occurring in Algorithm~\ref{alg:approxnash}, where $x'$ results from $x$ by a deviation in the while-loop of player~$i$ from $x_i$ to $x_i'$. Then, if a player~$j \in N$ wants to deviate to a resource $r \in R$ with $\ell_{r}(x')=\ell_{x_i}(x')$, she prefers $x_i$. 
\end{applemma}
\begin{proof}
Note that since player~$i$ deviates from $x_i$, we know that $\ell_{x_i}(x)>\ell_{x_i}(x')=\ell_{r}(x') \geq \ell_{r}(x)$ for $r \ne x_i$. Thus, since the loads are decreasing along the resources, $a_{x_i}\leq a_{r}$ holds. Consequently, $\dev(j,x_i,x')\leq \dev(j,r,x')$ for $r \in R$ with $\ell_{x_i}(x')=\ell_{r}(x')$. Furthermore, the resource $x_i$ has a smaller index than the resource $r$ since the loads are decreasing. Thus, by construction, player~$j$ prefers the resource $x_i$.
\end{proof}

\begin{applemma}\label{lemma_equal_load}
Let $x$ and $x'$ be two strategy profiles occurring in Algorithm~\ref{alg:approxnash}, where $x'$ results from $x$ either due to the addition of some new player, or by a deviation of a player in the while-loop.  
Denote the new or deviating player by~$i$.    
Then any player~$j$ with $\ell_{x_j'}(x')=\ell_{x_i'}(x')$ is happy with respect to $x'$. 
\end{applemma}
\begin{proof}
Note that since the loads are decreasing along the resources, $a_{x_j'}\leq a_{x_i'}$ holds. Therefore, if $x_j'\neq x_i'$, player~$j$ does not want to deviate to $x_i'$ (her cost can only increase).  
Now let $r\neq x_j'$ be some possible alternative for player~$j$. We show $c_{x_j'}(x')\leq \dev(j,r,x')$, that is, player~$j$ does not want to deviate to $r$. Due to the above, we may assume that $r\neq x_i'$. 
First note that $c_{x_j'}(x')\leq c_{x_i'}(x')$ holds since $a_{x_j'}\leq a_{x_i'}$. 
Secondly, $c_{x_i'}(x')\leq \dev(i,r,x')$ since $x_i'$ was a best response for player~$i$. 
Finally, $\dev(i,r,x')=\dev(j,r,x')$ holds since the deviation cost can only be different for two players if they have different load (see Lemma~\ref{lemma_deviation}). 
Altogether, we conclude $c_{x_j'}(x')\leq \dev(j,r,x')$ and thus player~$j$ does not want to deviate to $r$. 
\end{proof}

We can now state and prove the Lemmas~\ref{lemma_cond1} and~\ref{lemma_cond2} which are the basis of the case distinction used in the proof of Theorem~\ref{theo_symm_and_linear_1}. 
\begin{applemma}\label{lemma_cond1}
Assume that $i\in \{1,\ldots,k-1\}$ is a player who becomes unhappy due to the addition of the new player~$k$, that is, player~$i$ has a $K$-improving deviation $r$ with respect to the strategy profile $x':=(x_k',x_{-k})$ which results from $x$ due to the addition of the new player~$k$. 
Then, regarding the loads of the resources $x_k', x_i'=x_i$, and $r$, one of the cases shown in Table~\ref{table_after_k_is_placed}  needs to hold, where $M:=M(x)=\max\{\ell_{r'}(x): r'\in R\}$ denotes the maximum load with respect to $x$, that is, before player~$k$ is added. 
Also note that $x_i\neq x_k'$ holds since all players using $x_k'$ are happy with respect to $x'$. 
\end{applemma}
\begin{table*}[h]
\centering
\begin{tabular}{l c@{\hspace{0.1cm}} l c@{\hspace{0.1cm}} l c@{\hspace{0.1cm}} l}
\toprule
$\ell_{x_k'}(x)$ && $\ell_{x_i}(x')=\ell_{x_i}(x)$ && $\ell_r(x')$ && \text{further conditions} \\\midrule
$M-2$ && $M$&& $M-2$ && $x_i$ is the only resource with load $M$ in $x$ (also in $x'$) \\\hline\addlinespace[0.5em]
$M-1$ && $\leq M-1$ &&$M-1$  \\\hline\addlinespace[0.5em]
\multirow[t]{2}{*}{$M$} && $\leq M$ && $M$\\\addlinespace[0.25em]\cline{2-7}\addlinespace[0.25em]
&& $\leq M-1$&& $M-1$\\
\bottomrule 
\end{tabular}
\caption{Cases for the loads of the resources $x_k',x_i$ and $r$, where $r$ is a $K$-improving deviation for player~$i$ who becomes unhappy due to the addition of new player~$k$ on $x_k'$ (each row corresponds to one possible case).}
\label{table_after_k_is_placed}
\end{table*}
\begin{proof}
Let $i\in\{1,\ldots,k\}$ be a player who wants to change from her current strategy $x_i$ to a resource $r\neq x_i$.  

We first show that $c_{x_i}(x')\leq c_{x_i}(x)$, that is, the addition of player~$k$ to the game did not increase the cost experienced by player~$i$. 
Clearly, $\ell_{x_i}(x')=\ell_{x_i}(x)$, since the load on resources different from $x_k'$ has not changed. Therefore, the congestion-part of the cost experienced by player~$i$ has not changed. We now turn to the adversary part and show that it can only be smaller than before player~$k$ was added: Assume, by contradiction, that the adversary part has increased. In particular, the adversary part with respect to $x'$ has to be positive. This implies that $M:=M(x)=\max\{\ell_{r'}(x):r'\in R\}\geq \ell_{x_i}(x)=\ell_{x_i}(x')=\max\{\ell_{r'}(x'):r'\in R\}\geq M$ holds. Thus $x_i$ is a maximum load resource both before and after player~$k$ was added, and clearly the number of resources with load~$M$ can not decrease by the addition of player~$k$. This contradicts our assumption that the adversary part experienced by player~$i$ has increased, and altogether shows that the total cost of player~$i$ can only be smaller than before player~$k$ was added. 

Before player~$k$ was added, player~$i$ was happy with her strategy~$x_i$, thus she in particular did not want to change to resource~$r$ then. Furthermore, as we already showed, player~$i$'s own cost has not increased due to the addition of player~$k$. Since we assume that player~$i$ now wants to change to $r$, we conclude that the deviation cost of player~$i$ for $r$ has decreased, that is, $\dev(i,r,x')<\dev(i,r,x)$ holds. 
Clearly, the congestion part of the deviation cost did not decrease, since $\ell_r(x')\geq \ell_r(x)$ holds. Therefore, the adversary part needs to be smaller than before, and this is only possible if it was positive with respect to $x$. We thus get $\ell_r(x)\geq M-2$, where $\ell_r(x)=M-2$ is only possible if $x_i$ is the only resource with load $M$ in $x$. 
Furthermore, the adversary part of the deviation cost with respect to $x'$ needs to be smaller than $B$. Moreover $r=x_k'$ is not possible since the adversary part of the deviation cost needs to decrease. Therefore, we can assume in the following that $r\neq x_k'$ and consequently $\ell_r(x)=\ell_r(x')$.  
We now distinguish between the three possible values for $\ell_r(x)\in\{M-2,M-1,M\}$. 
If $\ell_r(x)=M$, then the adversary part of the deviation cost with respect to $x'$ can only be smaller than $B$ if $\ell_{x_k'}(x)=M$. 
If $\ell_r(x)=M-1$, we get that $\ell_{x_k'}(x)\in \{M-1,M\}$ needs to hold. 
Finally if $\ell_r(x)=M-2$ and $x_i$ is the only resource with load $M$ in $x$, we conclude $\ell_{x_k'}(x) =M-2$ (note that $x_i\neq x_k'$, and players having load $\ell_{x_k'}(x)+1$ are happy according to Lemma~\ref{lemma_equal_load}). 
 Using that $\ell_{r}(x)\neq \ell_{x_i}(x)-1$ (see Lemma~\ref{lemma_load-1}), this yields the conditions stated in Table~\ref{table_after_k_is_placed}, completing the proof.  
 \end{proof}

\begin{applemma}\label{lemma_cond2}
Assume that player~$i\in \{1,\ldots,k\}$ deviates in some iteration of the while-loop of Algorithm~\ref{alg:approxnash}, and denote the strategy profiles before and after this deviation by $x$ and $x':=(x_{-i},x_i')$, respectively. 
Assume further that player~$j\in \{1,\ldots,k\}\setminus\{i\}$ becomes unhappy due to player~$i$'s deviation, that is, $j$ was happy with respect to $x$, but she has a $K$-improving deviation $r \neq x_j'=x_j$ with respect to $x'$. 
Then, regarding the loads of the resources $x_i,x_i', x_j$ and $r$, one of the cases given in Table~\ref{table_new_possible_deviations} has to hold, where 
$M:=M(x)=\max\{\ell_{r'}(x): r'\in R\}$ denotes the maximum load with respect to $x$, that is, before player~$i$ deviated. 
Also note that $x_j \notin \{x_i,x_i'\}$, since all players using $x_i$ are already unhappy with respect to $x$, and all players using $x_i'$ are happy with respect to $x'$. 
\end{applemma}
\begin{table*}[h]\begin{center}
\begin{tabular}{l c@{\hspace{0.1cm}} l  c@{\hspace{0.1cm}} l c@{\hspace{0.1cm}}l c@{\hspace{0.1cm}} p{7cm}}
\toprule
$\ell_{x_i}(x)$ && $\ell_{x_i'}(x)$ && $\ell_{x_j}(x')=\ell_{x_j}(x)$ && $\ell_{r}(x')$ && \text{further conditions} \\\midrule
\multirow[t]{6}{*}{$M$} && \multirow[t]{2}{*}{$M$} && $\leq M-1$ && $M-1$ $(x_i)$ && \\\addlinespace[0.25em]\cline{4-9}\addlinespace[0.25em]
&&&& $\leq M$ && $M$ && \\\addlinespace[0.25em] \cline{2-9}\addlinespace[0.25em]
&& \multirow[t]{2}{*}{$\leq M-2$} && $M$ && $M$ &&$x_i$ is not the only resource with load $M$ in $x$ \\\addlinespace[0.25em]\cline{4-9}\addlinespace[0.25em]
&& && $M$ && $\leq M-2$ &&$x_i$ is not the only resource with load $M$ in $x$ \\\addlinespace[0.25em]\cline{2-9}\addlinespace[0.25em]
&& \multirow[t]{2}{*}{$\leq M-3$} && $M-1$ && $M-1$ $(x_i)$ &&$x_i$ is the only resource with load $M$ in $x$ \\\addlinespace[0.25em]\cline{4-9}\addlinespace[0.25em]
&& && $M-1$ && $\leq M-3$ &&$x_i$ is the only resource with load $M$ in $x$ \\\addlinespace[0.25em]\cline{1-9}\addlinespace[0.25em]
\multirow[t]{3}{*}{$\leq M-1$} &&  \multirow[t]{2}{*}{$M$} && $\leq M$ && $M$ && \\\addlinespace[0.25em]\cline{4-9}\addlinespace[0.25em]
&&&& $\leq M-1$ && $M-1$ && \\\addlinespace[0.25em]\cline{2-9}\addlinespace[0.25em] 
&& $M-1$ && $\leq M-1$ && $M-1$ && \\\addlinespace[0.25em]\cline{1-9}\addlinespace[0.25em]
\multirow[t]{2}{*}{$\leq M-2$} && \multirow[t]{2}{*}{$M-2$} && $M$ && $\leq M-3$ $(x_i)$ && $x_j$ is the only resource with load $M$ in $x$ \\\addlinespace[0.25em]\cline{2-9}\addlinespace[0.25em]
&&&& $M$ && $ M-2$ &&$x_j$ is the only resource with load $M$ in $x$ \\
\bottomrule 
\end{tabular}
\caption{Cases for the loads of the resources $x_i,x_i', x_j$ and $r$, where $r$ is a $K$-improving deviation for player~$j$ who becomes unhappy due to player~$i$'s deviation from $x_i$ to $x_i'$. Each row corresponds to one possible case. 
In some cases, we derive that if player~$j$ is chosen as the next deviating player in the while-loop, she deviates to $x_i$; and whenever this is the case,  $x_i$ is added in brackets.}
\label{table_new_possible_deviations}
\end{center}\end{table*}
 \begin{proof}
 Assume that player~$i$ is unhappy with her strategy $x_i$ in $x$ and changes to $x_i'$. Denote the resulting strategy profile by $x':=(x_{-i},x_i')$. Assume that player~$j$ is happy with her strategy $x_j$ in $x$ and wants deviate now, with respect to $x'$, to $r \ne x_j$. Note that $x_j \ne x_i$ since player~$j$ is happy with her strategy in $x$. Furthermore
$x_j = x_j'$ since only player~$i \ne j$ deviates in $x$. Additionally $x_j' \ne x_i'$ because player~$i$ is happy in $x'$. If player~$j$ is happy with respect to the profile $x$ and wants to deviate now in $x'$ from $x_j=x_j'$ to a resource $r \ne x_i$, then either her cost must be increase, that is $c_{x_j}(x')>c_{x_j}(x)$ or her deviation cost for $r$ must decrease, that is $\dev(j,r,x')<\dev(j,r,x)$. These two cases are discussed in Claim~\ref{claim_cost_increase} and Claim~\ref{claim_deviation_decrease}. The situations where a player~$j$ wants to deviate to $x_i$ in $x'$ are analyzed in Claim~\ref{claim_deviation_to_x_i}.

\begin{claim}\label{claim_cost_increase}
$c_{x_j}(x')>c_{x_j}(x) \ \Leftrightarrow$ One of the two cases displayed in Table~\ref{table_cost_increase} needs to hold, where $M:=M(x)=\max\{\ell_{r'}(x): r'\in R\}$ denotes the maximum load with respect to $x$. 
\begin{table*}[h]
\centering
\begin{tabular}{l c@{\hspace{0.5cm}} l c@{\hspace{0.5cm}} l c@{\hspace{0.5cm}}  p{7cm}}
\toprule
$\ell_{x_i}(x)$ && $\ell_{x_i'}(x)$ && $\ell_{x_j}(x)$ && further conditions \\\midrule
\multirow[t]{2}{*}{$M$} && \multirow[t]{2}{*}{$\leq M-2$} && $M$ &&  \\
\addlinespace[0.25em]\cline{5-7}\addlinespace[0.25em]
&&&& $M-1$ && $x_i$ is the only resource with load $M$ in $x$ \\
\bottomrule 
\end{tabular}
\caption{Conditions for a player~$j$ who is happy with his strategy in $x$ and wants to change from $x_j$ to $r$ considering the profile $x'$, which results from $x$ after player~$i$ moved from $x_i$ to $x_i'$ (each row corresponds to one possible case).}
\label{table_cost_increase}
\end{table*}
\end{claim}
\begin{proofClaim}[Proof of Claim~\ref{claim_cost_increase}]
Assume that player~$i$ changes from the only resource with load $M$ in $x$ to $x_i'$ with $\ell_{x_i'}(x) \le M-2$ and player~$j$ is on a resource $x_j'=x_j$ with $\ell_{x_j}(x)=M-1$. Since $x_j \neq x_i$, $x_j \neq x_i'$ (see above) we know $\ell_{x_j}(x')=M-1$. Then, player~$j$'s cost increases from $c_{x_j}(x)=a_{x_j}\ell_{x_j}(x)$ to $c_{x_j}(x')=a_{x_j}\ell_{x_j}(x)+\frac{B}{\tilde{p}}$, where $\tilde{p} \geq 2$ is the number of resources with load $M-1$ in $x'$. Assume now, player~$i$ changes from a resource $x_i \ne x_j$  with load $M$ to $x_i'$ with $\ell_{x_i'}(x) \le M-2$ and player~$j$ is on a resource $x_j$ with $\ell_{x_j}(x)=M$. Then, player~$j$'s cost increases from $c_{x_j}(x)=a_{x_j}\ell_{x_j}(x)+\frac{B}{p}$ to $c_{x_j}(x')=a_{x_j}\ell_{x_j}(x)+\frac{B}{p-1}$, where $p:=|M^{-1}(x)| \geq 2$ is the number of resources with load $M$ in $x$.

Now, we will show that if $c_{x_j}(x')>c_{x_j}(x)$, one of the two cases displayed in Table~\ref{table_cost_increase} needs to hold. Since $x_j=x_j'$, $x_j \ne x_i$ and $x_j \ne x_i'$ the congestion-part of the cost experienced by player $j$ does not change, comparing the profile $x'$ with $x$. Thus, if $c_{x_j}(x')>c_{x_j}(x)$ holds, the adversary-part must be increase. Therefore, the adversary-part with respect to $x'$ can not be $0$, which implies $x_j$ has maximum load in $x'$.  

If $x_j$ has not maximum load in $x$, then $x_i$ was the only resource with maximum load in $x$, which is $M$ and deviates to $x_i'$ with $\ell_{x_i'}(x) \leq M-2$, $\ell_{x_j}(x)=M-1$.

If $x_j$ has maximum load in $x$, which is $M$, the number of resources with maximum load must be decrease since the adversary-part must be increase. This implies that player $i$ deviates from $x_i$ with $\ell_{x_i}(x)= M$ to $x_i'$ with $\ell_{x_i'}(x) \leq M-2$.
\end{proofClaim}

\begin{claim}\label{claim_deviation_decrease}
$\dev(j,r,x')<\dev(j,r,x) \text{ for } r \ne x_i \Leftrightarrow r \ne x_i'$ and one of the four cases displayed in Table~\ref{table_deviation_decrease} needs to hold, where $M:=M(x)=\max\{\ell_{r'}(x): r'\in R\}$ denotes the maximum load with respect to $x$.
\begin{table*}[h]
\centering
\begin{tabular}{l c@{\hspace{1cm}} l c@{\hspace{1cm}} l c@{\hspace{1cm}} p{6.5cm}}
\toprule
$\ell_{x_i}(x)$  &&  $\ell_{x_i'}(x)$ && $\ell_{r}(x)$ && \text{further conditions} \\\midrule
\multirow[t]{2}{*}{$\leq M$} && \multirow[t]{2}{*}{$M$} &&  $M$ &&\\
\addlinespace[0.25em]\cline{5-7}\addlinespace[0.25em]
&&&&$M-1$ &&\\
\addlinespace[0.25em]\hline\addlinespace[0.5em]
\multirow[t]{2}{*}{$\leq M-1$} && \multirow[t]{2}{*}{$M-1$} && $M-1$ &&\\
\addlinespace[0.25em]\cline{5-7}\addlinespace[0.25em]
&&&&$M-2$ && $x_j$ is the only resource with load $M$ in $x$\\
\addlinespace[0.25em]\hline\addlinespace[0.5em]
$\leq M-2$ && $M-2$ && $M-2$ && $x_j$ is the only resource with load $M$ in $x$ \\
\bottomrule 
\end{tabular}
\caption{Conditions for a player~$j$ who is happy with his strategy in $x$ and wants to change from $x_j$ to $r$ considering the profile $x'$, which results from $x$ after player~$i$ moved from $x_i$ to $x_i'$ (each row corresponds to one possible case).}
\label{table_deviation_decrease}
\end{table*}
\end{claim}
\begin{proofClaim}[Proof of Claim~\ref{claim_deviation_decrease}]
Assume that player~$i$ changes from $x_i$ with $\ell_{x_i}(x) \leq M$ to $x_i'$ with $\ell_{x_i'}(x) = M$. Then, with respect to the profile $x'$, a deviation to a resource $r \ne x_i, x_i'$ with $\ell_{r}(x) = M$ would now be better for player~$j$. That is $\dev(j,r,x')=a_{r}(\ell_{r}(x)+1)+\frac{B}{2}<a_{r}(\ell_{r}(x)+1)+B=\dev(j,r,x)$. Let us consider now the situation where player~$i$ changes from a resource $x_i$ with load $\leq M$ to $x_i'$ with $\ell_{x_i'}(x)= M$, or from a resource $x_i$ with load $\leq M-1$ to $x_i'$ with $\ell_{x_i'}(x)= M-1$ and $r \ne x_i, x_i'$ is a resource with load $M-1$. Then, player~$j$'s deviation cost for $r$ decreases from $\dev(j,r,x)=a_{r}(\ell_{r}(x)+1)+\frac{B}{p+1}$ to $\dev(j,r,x')=a_{r}(\ell_{r}(x)+1)$, or from $\dev(j,r,x)=a_{r}(\ell_{r}(x)+1)+\frac{B}{p+1}$ to $\dev(j,r,x')=a_{r}(\ell_{r}(x)+1)+\frac{B}{p+2}$ respectively, where $p:=|M^{-1}(x)|$ is the number of resources with load $M$ in $x$. If player~$i$ changes from $x_i$ with load $ \leq M-1$ to $x_i'$ with load $M-1$, or from $x_i$ with load $ \leq M-2$ to $x_i'$ with load $M-2$ and player~$j$ is on the only resource with load $M$ in $x$, then a deviation to $r \ne x_i, x_i'$ with $\ell_{r}(x)= M-2$ is now better than before. In the first mentioned case we have $\dev(j,r,x')=a_{r}(\ell_{r}(x)+1)<a_{r}(\ell_{r}(x)+1)+\frac{B}{\tilde{p}+2}=\dev(j,r,x)$; in the second case we have $\dev(j,r,x')=a_{r}(\ell_{r}(x)+1)+\frac{B}{\tilde{p}+3}<a_{r}(\ell_{r}(x)+1)+\frac{B}{\tilde{p}+2}=\dev(j,r,x)$, where $\tilde{p}$ is the number of resources with load $M-1$ in $x$.

Now we will show that if $\dev(j,r,x')<\dev(j,r,x)$, one of the two cases displayed in Table~\ref{table_deviation_decrease} needs to hold and $r \ne x_i'$. First, we discuss the case $r=x_i'$. Since player $i$ changes from $x_i$ to $x_i'$ in $x$ resulting the profile $x'$, we thus get $\ell_r(x)+1=\ell_r(x')$. This implies that the congestion part of the deviation cost to $r=x_i'$ does not decrease. The same holds for the adversary part since the load from $x_i'$ increases while the load of the other resources do not increase. This implies $\dev(j,x_i',x') \geq \dev(j,x_i',x)$ which is a contraction to our assumption.

Assume now that $r \ne x_i'$. Since $r \ne x_i$ we get $\ell_{r}(x)=\ell_{r}(x')$ which implies that the congestion part of the deviation cost to $r$ is the same in $x$ and $x'$. Thus, the adversary part must decrease comparing $x$ to $x'$ to achieve $\dev(j,r,x')<\dev(j,r,x)$. Therefore, the adversary part of the deviation cost needs to be positive with respect to $x$. 
Consequently, $\ell_r(x)+1=M(x_{-j},r)=:M_{j,r} \in \{M-1,M,M+1\}$ needs to hold. 
We now distinguish between the three possible values for $M_{j,r}$.

Case $M_{j,r}=M+1$: This implies $\ell_r(x)=M$. If player~$j$ changes to $r$ in $x$ she has to pay the whole adversary budget $B$ (since $r$ is the only resource with load $M+1$ after this deviation). Now, with respect to $x'$, the adversary part of the deviation cost must decrease. Therefore, there must be another resource with load $M+1$ in $x'$. This implies that player~$i$'s move must be from a resource with load $\leq M$ to a resource with load $M$.

Case $M_{j,r}=M$: This implies $\ell_r(x) = M-1$ since we assume $\ell_r(x)+1=M_{j,r}=M$ (see above). The adversary part of the deviation cost to $r$ with respect to $x$ is either $\frac{B}{p}$ or $\frac{B}{p+1}$, where $p:=|M^{-1}(x)|$ is the number of resources with load $M$ in $x$. The first mentioned case occurs if $\ell_{x_j}(x)=M$, the second if $\ell_{x_j}(x)<M$. Now, with respect to $x'$, the adversary part needs to be smaller. Thus, there are two possible situations. Either the adversary part gets $0$ or $\frac{B}{\bar{p}}$, where $\bar{p}$ is the number of resources with load $M_{j,r}=M$ in $(x_{-j}',r)$. If the the adversary part gets $0$ we know with $\ell_r(x')+1=\ell_r(x)+1=M$ that the maximum load in $x'$ is $>M$. Thus, player~$i$'s move must be from a resource with load $\leq M$ to a resource with load $M$. If the adversary part gets $\frac{B}{\bar{p}}$ we know with the assumption that the adversary part needs to decrease, that player~$i$'s move must be from a resource with load $\leq M-1$ to a resource with load $M-1$. 

Case $M_{j,r}=M-1$: This implies that $x_j$ must be the only resource with load $M$ in $x$ and $\ell_r(x) \leq M-2$. Since we assumed that $\ell_r(x)+1=M_{j,r}=M-1$ (see above), we know $\ell_r(x)=M-2$. The adversary part of the deviation cost for $r$ with respect to $x$ is $\frac{B}{\tilde{p}+2}$, where $\tilde{p}$ is the number of resources with load $M-1$ in $x$. Note that $x_j$ and $r$ have load $M-1$ after the change. Now, with respect to $x'$, the adversary part needs to be smaller.
Thus there are two possible situations. Either the adversary part gets $0$ or $\frac{B}{\hat{p}}$, where $\hat{p}$ is the number of resources with load $M_{j,r}=M-1$ in $(x_{-j}',r)$. If the the adversary part gets $0$ we know with $\ell_r(x')+1=M-1$ that the maximum load in $x'$ is $>M-1$. Thus, player~$i$'s move must be from a resource with load $\leq M-1$ to a resource with load $M-1$. Note that $x_i$ and $x_i'$ have maximum load $M-1$ because $x_j \ne x_i,x_i'$ (see above) is the only resource with load $M$ in $x$. If the adversary part gets $\frac{B}{\hat{p}}$ we know with the assumption that the adversary part needs to decrease, that player~$i$'s move must be from a resource with load $\leq M-2$ to a resource with load $M-2$. 
\end{proofClaim}

\begin{claim}\label{claim_deviation_to_x_i}
Assume player~$i$ changed in $x$ from $x_i$ to $x_i'$ resulting the profile $x'$. Then, with respect to $x'$, a player~$j$ might want to deviate to $x_i$ in $x'$ if one of the six cases displayed in Table~\ref{table_deviation_to_x_i} holds, where $M:=M(x)=\max\{\ell_{r'}(x): r'\in R\}$ denotes the maximum load with respect to $x$. 
\begin{table*}[h]
\centering
\begin{tabular}{l c@{\hspace{1cm}} l c@{\hspace{1cm}} l c@{\hspace{1cm}} p{6.5cm}}
\toprule
$\ell_{x_i}(x)$ && $\ell_{x_i'}(x)$ && $\ell_{x_j}(x)$ && \text{further conditions} \\\midrule
\multirow[t]{2}{*}{$M$}&& $M$ && $\leq M$ && \\
\addlinespace[0.25em]\cline{3-7}\addlinespace[0.25em]
&& $M-1$ && $\leq M-1$ && \\\addlinespace[0.25em]\cline{3-7}\addlinespace[0.25em]
&& \multirow[t]{2}{*}{$\leq M-2$} &&  $M$ &&\\
\addlinespace[0.25em]\cline{5-7}\addlinespace[0.25em]
&&&&$M-1$ && $x_i$ is the only resource with load $M$ in $x$\\
\addlinespace[0.25em]\cline{1-7}\addlinespace[0.25em]
\multirow[t]{2}{*}{$\leq M-1$} && $M-1$ && $M$ && \\
\addlinespace[0.25em]\cline{3-7}\addlinespace[0.25em]
 && $M-2$ && $M$ && $x_j$ is the only resource with load $M$ in $x$ \\
\bottomrule 
\end{tabular}
\caption{Conditions for a player~$j$ who wants to change to $x_i$ considering the profile $x'$, which results from $x$ after player~$i$ moved from $x_i$ to $x_i'$ (each row corresponds to one possible case).}
\label{table_deviation_to_x_i}
\end{table*}
\end{claim}
\begin{proofClaim}[Proof of Claim~\ref{claim_deviation_to_x_i}]
Obviously, player~$j$ wants to change from her current strategy $x_j$ to $x_i$ in $x'$ if and only if 
\begin{equation}\label{deviation_j}
c_{x_j}(x')> K\dev(j,x_i,x').
\end{equation}
Moreover, player~$i$ wants to change from $x_i$ to $x_i'$ with respect to $x$ if and only if 
\begin{equation}\label{deviation_i}
c_{x_i}(x)> K\dev(i,x_i',x).
\end{equation}
We first show by contradiction that it is not possible that all of the three following inequalities are fulfilled: $\dev(j,x_i,x') \geq c_{x_i}(x)$, $\dev(i,x_i',x) \geq \dev(j,x_i',x)$, and $c_{x_j}(x) \geq c_{x_j}(x')$. Assume that these three inequalities are fulfilled. Together with  inequality~\eqref{deviation_j} and inequality~\eqref{deviation_i}, we thus get:
\begin{align*}
c_{x_j}(x) &\geq c_{x_j}(x')> K \dev(j,x_i,x') \geq K c_{x_i}(x)\\
&> K^2 \dev(i,x_i',x) \geq K^2 \dev(j,x_i',x).
\end{align*}
With $K^2 \dev(j,x_i',x)> K \dev(j,x_i',x)$ we can conclude that player~$j$ wants to deviate to $x_i'$ with respect to $x$ and $c_{x_j}(x)>Kc_{x_i}(x)>c_{x_i}(x)$. This contradicts our assumption that in the algorithm player~$i$ deviates in $x$ and not player~$j$. Thus, at least one of these three mentioned inequalities is not fulfilled. We now distinguish between which of these inequalities does not hold.

Case $\dev(j,x_i,x') < c_{x_i}(x)$: We will show that in this situation one of the two following cases holds. Either $\ell_{x_j}(x) \leq M$ and player~$i$ deviates from $x_i$ with load $M$ to $x_i'$ with load $M$ in $x$, or $\ell_{x_j}(x) \leq M-1$ and player~$i$ deviates from $x_i$ with load $M$ to $x_i'$ with load $M-1$ in $x$. Obviously, the congestion part of the cost experienced by player~$i$ on $x_i$ in $x$ and the congestion part of the deviation cost experienced by player~$j$ by changing to $x_i$ in $x'$ is the same since $\ell_{x_i}(x)=\ell_{x_i}(x_{-j}',x_i)$. Thus, to achieve $\dev(j,x_i,x') < c_{x_i}(x)$, the adversary part needs to be smaller. This implies that the adversary part of the cost experienced by player~$i$ on $x_i$ in $x$ needs to be positive. Therefore, $\ell_{x_i}(x)=M$ and the adversary part of the cost is $\frac{B}{p}$, where $p:=|M^{-1}(x)|$ is the number of resources with load $M$ in $x$. Since the adversary part of the deviation cost experienced by player~$j$ by changing to $x_i$ in $x'$ must be smaller, it can be $0$ or $\frac{B}{\tilde{p}}$, where $\tilde{p}$ is the number of resources with load $M$ in $(x_{-j}',x_i)$, and $\tilde{p}>p$. If the adversary part is $0$, we know with $\ell_{x_i}(x)=\ell_{x_i}(x_{-j}',x_i)=M$ that the maximum load is $M+1$ in $(x_{-j}',x_i)$. This implies that player~$i$ changes from $x_i$ with load $M$ to $x_i'$ with load $M$ in $x$. Now we discuss the case that the adversary part is $\frac{B}{\tilde{p}}$. In this situation player~$i$ changes from $x_i$ with load $M$ to $x_i'$ with load $ \leq M-1$ in $x$. Let $p'$ be the number of resources with load $M$ in $x'$. If $\ell_{x_j}(x)= M$ we get $\tilde{p}=p' \leq p$ which contradicts $\tilde{p}>p$. Thus, $\ell_{x_j}(x) \leq M-1$. If $\ell_{x_i'}(x) \leq M-2$ we get $p=p'+1$, and with $\ell_{x_j}(x) \leq M-1$ we get $\tilde{p}=p'+1$. Together we get $p=p'+1=\tilde{p}$ which contradicts $\tilde{p}>p$. Thus, $\ell_{x_i'}(x)= M-1$.

Case $\dev(i,x_i',x) < \dev(j,x_i',x)$: We show that in this situation one of the two following cases holds. Either $\ell_{x_j}(x) = M$ and player~$i$ deviates from $x_i$ with load $\leq M-1$ to $x_i'$ with load $M-1$ in $x$, or $x_j$ is the only resource with load $M$ in $x$ and player~$i$ deviates from $x_i$ with load $\leq M-1$ to $x_i'$ with load $M-2$ in $x$. According to Lemma~\ref{lemma_deviation} we know that either $\ell_{x_i'}(x)=M-1$ and one of the two resources $x_j$, $x_i$ has load $M$ while the other has load $<M$ in $x$, or $\ell_{x_i'}(x)=M-2$ and one of the two resources $x_j$, $x_i$ is the only resource with load $M$ in $x$. 
To make that more precise, we have to know the relation between the load of $x_j$ and $x_i$ in $x$. We prove by contradiction that $\ell_{x_j}(x) \geq \ell_{x_i}(x)$. Assume that $\ell_{x_j}(x) < \ell_{x_i}(x)$. This implies $\ell_{x_i}(x)=M$ and $\ell_{x_j}(x) \leq M-1$ with Lemma~\ref{lemma_deviation}. Furthermore, we know $\ell_{x_i'}(x) \in  \{M-1,M-2\}$ according to Lemma~\ref{lemma_deviation}. If $\ell_{x_i'}(x) = M-1$ we get $\dev(i,x_i',x)=a_{x_i'}M+\frac{B}{p} > a_{x_i'}M+\frac{B}{p+1}=\dev(j,x_i',x)$, where $p:=|M^{-1}(x)|$ is the number of resources with load $M$ in $x$. This contradicts our assumption that $\dev(i,x_i',x) < \dev(j,x_i',x)$. If $\ell_{x_i'}(x) = M-2$ and $x_i$ is the only resource with load $M$ in $x$, we get $\dev(i,x_i',x)=a_{x_i'}(M-1)+\frac{B}{\tilde{p}+2} > a_{x_i'}(M-1)=\dev(j,x_i',x)$, where $\tilde{p}$ is the number of resources with load $M-1$ in $x$. This contradicts our assumption, too. Thus, $\ell_{x_j}(x) \geq \ell_{x_i}(x)$ needs to hold and with Lemma~\ref{lemma_deviation} we know that either $\ell_{x_i'}(x)=M-1$, $\ell_{x_j}(x)=M$ and $\ell_{x_i}(x) \leq M-1$, or $\ell_{x_i'}(x)=M-2$ and $x_j$ is the only resource with load $M$ in $x$.

Case $c_{x_j}(x)<c_{x_j}(x')$: We show that in this situation one of the following two cases holds. Either $\ell_{x_j}(x) = M-1$ and player~$i$ deviates from $x_i$ which is the only resource with load $M$ in $x$ to $x_i'$ with load $ \leq M-2$ in $x$, or $\ell_{x_j}(x)= M$ and player~$i$ deviates from $x_i$ with load $M$ to $x_i'$ with load $\leq M-2$ in $x$. 
Clearly, the congestion part of the cost experienced by player~$j$ in $x$ and $x'$ does not change since $\ell_{x_j}(x)=\ell_{x_j}(x')$. Note that $x_j \neq x_i$ and $x_j \neq x_i'$ (see above). Thus, the adversary part of player~$j$'s cost must increase comparing $x'$ to $x$. This implies that the adversary part of player~$j$'s cost is positive with respect to $x'$. Therefore, the maximum load in $x'$ must be $\leq M$, since $\ell_{x_j}(x')=\ell_{x_j}(x) \leq M$. In particular, $\ell_{x_i'}(x') \leq M$ and thus $\ell_{x_i'}(x) \leq M-1$ needs to hold. Moreover $\ell_{x_j}(x') \in \{M-1,M\}$ is the maximum load in $x'$. We now distinguish between the possible values for the maximum load in $x'$. 
Assume first that $\ell_{x_j}(x')=M-1$. This implies that player~$i$ changes from $x_i$ which is the only resource with load $M$ in $x$ to $x_i'$ with load $ \leq M-2$ in $x$. Thus, the adversary part of player~$j$'s cost is $0$ with respect to $x$ and $\frac{B}{\tilde{p}}$ with respect to $x'$, where $\tilde{p} \ge 2$ is the number of resources with load $M-1$ in $x'$. This shows an increase of the adversary part.
Assume now that $\ell_{x_j}(x')=M$. Thus, the adversary part of player~$j$'s cost is $\frac{B}{p}$ with respect to $x$ and $\frac{B}{p'}$ with respect to $x'$, where $p$ and $p'$ denote the numbers of resources with load $M$ in $x$ and $x'$, respectively. Since the adversary part needs to increase we know $p>p'$. Therefore, player~$i$ deviates from $x_i$ with load $M$ to $x_i'$ with load $\leq M-2$ in $x$.
\end{proofClaim}

Altogether, we can conclude the conditions for a player~$j$ who is happy with her strategy in $x$ and wants to change from $x_j'=x_j$ to $r$ considering the profile $x'$, which results from $x$ after player~$i$ moved from $x_i$ to $x_i'$. These conditions are displayed in Table~\ref{table_new_possible_deviations_without_lemma}.

\begin{table*}[h]
\centering
\begin{tabular}{l c@{\hspace{0.1cm}} l  c@{\hspace{0.1cm}} l c@{\hspace{0.1cm}}l c@{\hspace{0.1cm}} p{7cm}}
\toprule
$\ell_{x_i}(x)$ && $\ell_{x_i'}(x)$ && $\ell_{x_j}(x')=\ell_{x_j}(x)$ && $\ell_{r}(x')$ && \text{further conditions} \\\midrule
\multirow[t]{13}{*}{$M$} && \multirow[t]{3}{*}{$M$} && $\leq M$ && $M$ && \\\addlinespace[0.25em]\cline{4-9}\addlinespace[0.25em]
&&&& $\leq M$ && $M-1$ && \\\addlinespace[0.25em] \cline{4-9}\addlinespace[0.25em]
&&&& $\leq M$ && $M-1$ $(x_i)$ && \\\addlinespace[0.25em] \cline{2-9}\addlinespace[0.25em]
&&$M-1$ && $\leq M-1$ && $M-1$ $(x_i)$ && \\\addlinespace[0.25em]\cline{2-9}\addlinespace[0.25em]
&& $M-2$ && \multirow[t]{4}{*}{$M$} && $\leq M$ &&$x_i$ is not the only resource with load $M$ in $x$ \\\addlinespace[0.25em]\cline{4-9}\addlinespace[0.25em]
&&&& $M$ && $\leq M-1$ $(x_i)$&&$x_i$ is not the only resource with load $M$ in $x$ \\\addlinespace[0.25em]\cline{4-9}\addlinespace[0.25em]
&&&& $M-1$ && $M-1$ $(x_i)$ &&$x_i$ is the only resource with load $M$ in $x$ \\\addlinespace[0.25em]\cline{4-9}\addlinespace[0.25em]
&&&& $M-1$ && $\leq M-1$ &&$x_i$ is the only resource with load $M$ in $x$ \\\addlinespace[0.25em]\cline{2-9}\addlinespace[0.25em]
&& \multirow[t]{4}{*}{$\leq M-3$} && $M$ && $\leq M$ &&$x_i$ is not the only resource with load $M$ in $x$ \\\addlinespace[0.25em]\cline{4-9}\addlinespace[0.25em]
&&&& $M$&& $M-1$ $(x_i)$ &&$x_i$ is not the only resource with load $M$ in $x$ \\\addlinespace[0.25em]\cline{4-9}\addlinespace[0.25em]
&&&&$M-1$ && $M-1$ $(x_i)$&& $x_i$ is the only resource with load $M$ in $x$\\\addlinespace[0.25em]\cline{4-9}\addlinespace[0.25em]
&&&& $M-1$&& $\leq M-1$&& $x_i$ is the only resource with load $M$ in $x$\\\addlinespace[0.25em]\cline{1-9}\addlinespace[0.25em] 
\multirow[t]{12}{*}{$\leq M-1$} && \multirow[t]{2}{*}{$M$} && $\leq M$ && $
M$&& \\\addlinespace[0.25em]\cline{4-9}\addlinespace[0.25em]
&&&& $\leq M$ && $M-1$ && \\\addlinespace[0.25em] \cline{2-9}\addlinespace[0.25em]
&&\multirow[t]{5}{*}{$M-1$} && $M$ && $M-2$ && $x_j$ is the only resource with load $M$ in $x$\\\addlinespace[0.25em]\cline{4-9}\addlinespace[0.25em]
&&&&$M$ && $\leq M-2$ $(x_i)$&&\\\addlinespace[0.25em] \cline{4-9}\addlinespace[0.25em]
&&&&$\leq M$ && $M-1$&& \\
\addlinespace[0.25em] \cline{4-9}\addlinespace[0.25em]
&&$M-2$ && $M$ && $\leq M-2$ $(x_i)$ && $x_j$ is the only resource with load $M$ in $x$\\\addlinespace[0.25em]\cline{1-9}\addlinespace[0.25em]
\multirow[t]{2}{*}{$\leq M-2$} && $M-2$ && $M$ && $M-2$ && $x_j$ is the only resource with load $M$ in $x$\\\addlinespace[0.25em] \cline{4-9}\addlinespace[0.25em]
&&&&$M$ && $\leq M-3$ $(x_i)$&& $x_j$ is the only resource with load $M$ in $x$\\
\bottomrule 
\end{tabular}
\caption{Conditions for a player~$j$ who is happy with her strategy in $x$ and wants to change from $x_j'=x_j$ to $r$ considering the profile $x'$, which results from $x$ after player~$i$ moved from $x_i$ to $x_i'$ (each row corresponds to one possible case). Whenever there is an explicit resource~$r$ where player~$i$ wants to deviate to, we add this resource in brackets.}
\label{table_new_possible_deviations_without_lemma}
\end{table*}

Using Lemma~\ref{lemma_equal_load}, Lemma~\ref{lemma_load-1} and Lemma~\ref{lemma_x_i_preference} we can reduce Table~\ref{table_new_possible_deviations_without_lemma} to Table \ref{table_new_possible_deviations}, completing the proof of Lemma~\ref{lemma_cond2}.
 \end{proof}
 
 \section{Omitted Proofs for the Main Results from Section~\ref{sec:unrestricted}}\label{appendix_maintheo}
 In this section, we provide complete proofs for Theorem~\ref{theo_symm_and_linear_1} and Corollary~\ref{cor_2}. 
 \begin{proof}[Proof of Theorem~\ref{theo_symm_and_linear_1}]
It suffices to show that in iteration~$k\geq 3$ of the for-loop, the while-loop terminates. Let $x$ and $x'$ denote the profiles directly before and after the new player~$k$ is added. 
If all players are happy with their strategy in $x'$, the statement follows; thus assume that player~$i$ changes from $x_i'=x_i$ to $r$ in the while-loop. 
Due to Lemma~\ref{lemma_cond1}, we know that one of the three cases displayed in Table~\ref{table_after_k_is_placed} needs to hold. We now analyze each of these cases, and make repeated use of Lemma~\ref{lemma_cond2}. 
To this end, let $M:=M(x)=\max\{\ell_{r'}(x):r'\in R\}$ be the maximum load in $x$, that is, before player~$k$ is added, and recall that $x_k'$ denotes the resource to which player~$k$ is added, thus $\ell_{x_{k}'}(x)$ denotes the corresponding load before player~$k$ is added. 
Figure~\ref{fig_tree} illustrates the complete case distinction that we carry out in the proof, and we start with the analysis of the three different cases regarding $\ell_{x_{k}'}(x)$.
 
\begin{figure*}[h]
\begin{center}
\begin{tikzpicture}[sibling distance=9em,
  every node/.style = {shape=rectangle, rounded corners,
    draw, align=center}]]
  \node {$\ell_{x_k'}(x)$}
    child { node {$M-2$} }
		child { node {$M-1$} }
    child { node {$M$}
      child { node {players on $r_1$ \\ never deviate} }
			child { node {player~$j$ deviates \\ from $r_1$ to $r'$}
        child { node {$\ell_{r'}(x^{t+1})=M-1$} }
        child { node {$r'=r^t$} 
				  child { node {next dev.: \\ $\leq M-1\rightarrow M-1$} } 
					child { node {next dev.: \\ $\bar{r}^t\rightarrow r_1$} }
					child { node {next dev.: \\ $\bar{r}^t\rightarrow r^{t-1}$} }
					}
        }
       };
\end{tikzpicture} 
\end{center}
\caption{Illustration for the case distinction in the proof of Theorem~\ref{theo_symm_and_linear_1}.}\label{fig_tree}
\end{figure*}
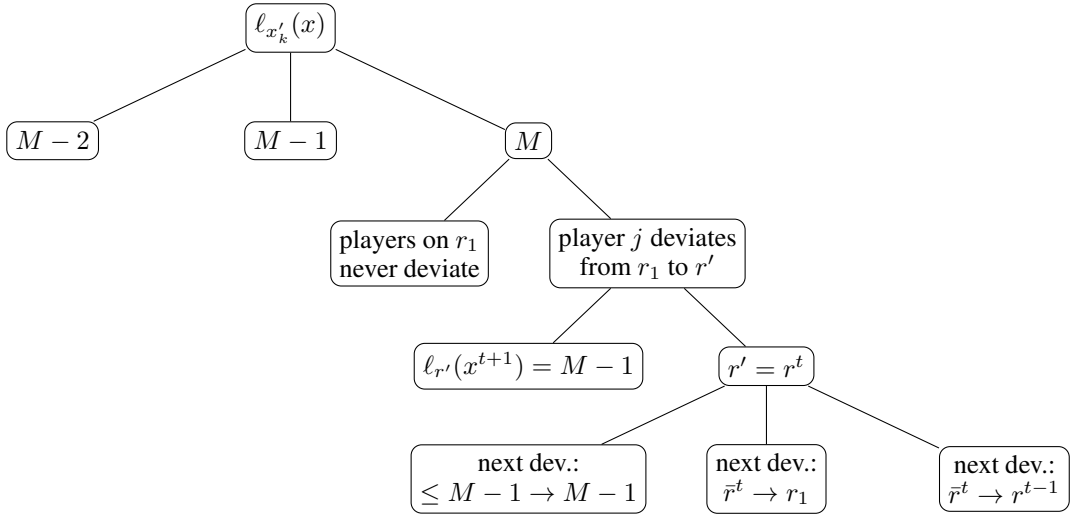

\begin{itemize}[leftmargin=*]
\item[] 
Case $M-2:$

If $\ell_{x_k'}(x)=M-2$, then $M(x')=M$ and $x_i=x_i'$ is the only resource with load $M$ in $x'$ and $\ell_{r}(x')=M-2$ holds.  
Using Lemma~\ref{lemma_cond2}, we conclude that there are no players who become unhappy due to player~$i$'s change, and the players on $x_i'=x_i$ (the only unhappy players w.r.t. $x'$) are now happy due to Lemma~\ref{lemma_equal_load}. Therefore, after player~$i$ deviated, all players are happy with their strategy and the while-loop terminates after one iteration. 

Case $M-1$:

If $\ell_{x_k'}(x)=M-1$, then $M(x')=M$, $\ell_{x_i}(x')\leq M-1$ and $\ell_{r}(x')=M-1$ hold. Lemma~\ref{lemma_cond2} yields that after this deviation, the only players who might be unhappy are players using resources  with load $\leq M-1$, and they want to change to resources with load $M-1$. 
But since there are at most $m-2$ resources with load $M-1$ in $x'$, we conclude that the while-loop terminates after $O(m)$ iterations.

Case $M$:

Now turn to the case that $\ell_{x'_k}(x)=M$. Note that $x'_k=r_1$ and $M(x')=M+1$ hold, and that $r_1$ is the only resource with load $M+1$ in $x'$. We have to consider the different possibilities regarding the loads of $x_i$ and $r$ as given in Table~\ref{table_after_k_is_placed}, namely that player~$i$ changes from load $\leq M$ to load~$M$ or from $\leq M-1$ to $M-1$. 
First note that if $\ell_{r}(x')=M$, the only players who might want to change in $x'':=(x'_{-i},r)$ are using resources with load~$\leq M$ and want to change to load~$M$, or want to change from load~$\leq M-1$ to load~$M-1$. 
This follows from Lemma~\ref{lemma_cond2}, where one should note that the maximum load in $x'$ is $M+1$. 
For the case that $\ell_{r}(x')=M-1$ (and consequently, $\ell_{x_i}(x')\leq M-1$), the only new unhappy players could be players using $x'_k=r_1$, and they might want to change to a resource with load~$M-1$, or to $x_i'$ (see Lemma~\ref{lemma_cond2}). 

Using the above, we can argue that if the players on $r_1$ never deviate during the while-loop, then the only changes are from resources having load $\leq M$ to $M$, or from $\leq M-1$ to $M-1$. But this terminates after $O(m)$ deviations (there can be at most $m-1$ deviations to a resource with load $M$, and at most $2m-2$ deviations to a resource with load $M-1$).

Thus we can now assume that in some iteration of the while-loop, a player~$j$ deviates from $r_1$ to a resource~$r'$.  
Consider the first such change, that is, all changes before have been from load $\leq M$ to $M$ or from $\leq M-1$ to $M-1$. More exactly, note that the changes before cannot include a change to load $M$, because after such a change, all players having load $M+1$ are happy (see Lemma~\ref{lemma_equal_load}) and never become unhappy afterwards (see Lemma~\ref{lemma_cond2} and note that $r_1$ is not the only resource with load~$M+1$ anymore). 
Therefore, all changes before were moves from load $\leq M-1$ to $M-1$ (and there was at least one such change). 
Let $i_1,\ldots,i_t$ be the corresponding sequence of deviating players, where $i_1=i$ is the first, and $i_t$ the last player deviating before player~$j$. 
Let $r^s$ and $\bar{r}^s$ denote player~$i_s$'s resources before and after her change, for $s=1,\ldots,t$, and let $x^{s+1}$ be the strategy profile resulting from $i_s$'s change.  
Finally, let $x^{t+2}:=(x_{-j}^{t+1},r')$ be the strategy profile resulting from $x^{t+1}$ due to player~$j$'s deviation from $r_1$ to $r'$. 
Figure~\ref{fig_3} shows an illustration for the situation before player~$j$ deviates. 

\begin{figure}[h]
\centering
\scalebox{0.8}{
\begin{tikzpicture}[decoration=brace]
\draw  (1, 0) rectangle (1.5, 4);
\draw[blue] (1, 4) rectangle (1.5, 4.5); 
\node[below] at (1.25, -0.2) {$r_1$};
\draw[decorate,  decoration={brace, mirror}]  (1,-1) -- node[below=0.6ex, text width=1cm]{load $M+1$}  (1.5,-1);
\path[->,blue] 
(1.25,4.5) edge[bend left=20]
node[above]{$j$} (4,5.5);
\draw  (2, 0) rectangle (2.5, 4); 
\draw  (3, 0) rectangle (3.5, 3.5); 
\draw[blue!20!black!30!green] (3, 3.5) rectangle (3.5, 4); 
\node[below] at (3.25, -0.1) {$\bar{r}^1$};
\draw  (4, 0) rectangle (4.5, 3.5); 
\draw[blue!20!black!30!green] (4, 3.5) rectangle (4.5, 4); 
\node[below] at (4.25, -0.1) {$\bar{r}^2$};
\draw  (5, 0) rectangle (5.5, 3.5); 
\draw[blue!20!black!30!green] (5, 3.5) rectangle (5.5, 4); 
\node[below] at (5.25, -0.1) {$\bar{r}^t$};
\draw[decorate,  decoration={brace, mirror}]  (2,-1) -- node[below=0.6ex]{load $M$}  (5.5,-1);
\draw  (6, 0) rectangle (6.5, 3.5); 
\draw  (7, 0) rectangle (7.5,3); 
\draw[blue!20!black!30!green,dashed] (7, 3) rectangle (7.5, 3.5); 
\node[below] at (7.25, -0.1) {$r^t$};
\path[->,blue!20!black!30!green] 
(9.25,3) edge[bend right=80]
node[above]{$i_1$} (3.25,4);
\draw  (8, 0) rectangle (8.5, 3); 
\draw[blue!20!black!30!green,dashed] (8, 3) rectangle (8.5, 3.5); 
\node[below] at (8.25, -0.1) {$r^2$};
\path[->,blue!20!black!30!green] 
(8.25,3.5) edge[bend right=50]
node[above]{$i_2$} (4.25,4);
\draw  (9, 0) rectangle (9.5, 2.5); 
\draw[blue!20!black!30!green,dashed] (9, 2.5) rectangle (9.5, 3); 
\node[below] at (9.25, -0.1) {$r^1$};
\draw  (10, 0) rectangle (10.5, 2.5); 
\path[->,blue!20!black!30!green] 
(7.15,3.5) edge[bend right=20]
node[above]{$i_t$} (5.5,4.05);
\draw[decorate,  decoration={brace, mirror}]  (6,-1) -- node[below=0.6ex]{load $\leq M-1$}  (10.5,-1);
\end{tikzpicture}}
\caption{Situation before player~$j$ deviates (with $t=3$).}%
\label{fig_3}%
\end{figure}
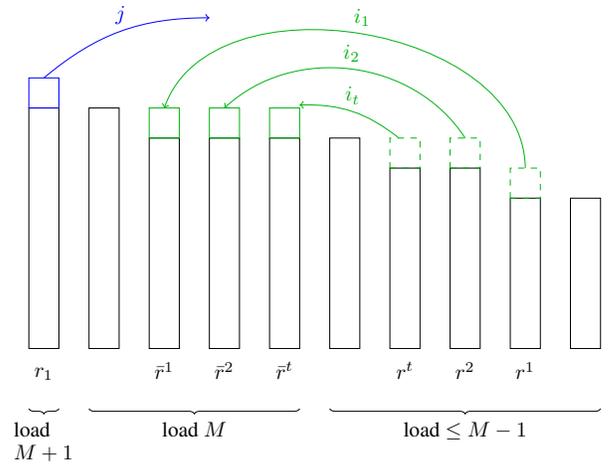

Note that $r_1$ is the only resource with load $M+1$ in $x^{t+1}$. Furthermore, all players having load $M$ are happy in $x^{t+1}$ (see Lemma~\ref{lemma_equal_load}), thus the only unhappy players in $x^{t+1}$ (except from the players using $r_1$) might want to change from load $\leq M-1$ to load $M-1$ or to load $M$. Now consider the situation after player~$j$ deviated. 
Recall that player~$j$ either deviated to a resource with load~$M-1$, or to $r^t$ (smallest resulting cost and smallest index among all resources where the players~$i_1,\ldots,i_t$ deviated from). 

Subcase $\ell_{r'}(x^{t+1})=M-1$:

If player~$j$ deviates to a resource with load~$M-1$, then all players having load $M$ are happy in $x^{t+2}$. Furthermore, there are no new unhappy players in $x^{t+2}$ (see Lemma~\ref{lemma_cond2}). Therefore, after player~$j$'s change, the only possible further changes are from load $\leq M-1$ to load $M-1$ or load $M$. 
Even more, we can also exclude changes from $\leq M-1$ to $M$, since the cost for deviating to $r_1$ (the best resource with load~$M$) would be the same as it was before the new player~$k$ was added to the game, and in $x$ no one wanted to change to $r_1$. Furthermore, the cost on a resource with load $\leq M-1$ can now only be smaller than it was in $x$. 
Altogether, the only further changes can be from load $\leq M-1$ to $M-1$, and this terminates after $O(m)$ iterations.

Subcase $r'=r^t$:

It remains to analyze the case that player~$j$ deviates from $r_1$ to $r^t$ (the resource where player~$i_t$ deviated from). 
New unhappy players might want to change from load~$M$ to $r_1$ or to a resource with load $\leq M-2$ (see Lemma~\ref{lemma_cond2}). 
Among these players, a player on $\bar{r}^t$ moves first (biggest index among the resources with load~$M$). 
Former unhappy players who are not included above want to change from $\leq M-1$ to $M-1$ or to $M$. More exactly, we can again exclude changes from $\leq M-1$ to $M$, since the cost for deviating to $r_1$ (the best resource with load $M$) would be the same as player~$j$ had to pay before her change, and if this is strictly smaller than the current cost of some unhappy player with load~$\leq M-1$, we get a contradiction to the fact that player~$j$ moved first (the other player could improve 'more' by deviating to $r^t$). 
Therefore, the next deviation may be from $\bar{r}^t$ to $r_1$, from $\bar{r}^t$ to $\leq M-2$, or from $\leq M-1$ to $M-1$.
 
In the last case, all players with load~$M$ are happy after the change, and the only further moves can be from load $\leq M-1$ to load~$M-1$. This terminates after $O(m)$ iterations.
 
The case $\bar{r}^t\rightarrow r_1$ (see Figure~\ref{fig_4} for an illustration) can be excluded since it leads to the following contradiction. 
Let $\ell_{r^t}$ denote the load of resource $r^t$ directly before player~$i_t$ deviated from it. Furthermore note that there are at least two resources with load~$M$ in $x^{t+2}$ (namely $r_1$ and $\bar r_1$). Then, the following three inequalities hold. 
 \begin{align}
a_{r^{t}}\cdot \ell_{r^t} &>K \cdot a_{\bar{r}^t}\cdot M \label{eq_2_proof_theo_symm_and_linear_1}\\
a_1\cdot(M+1)+B&>K \cdot  a_{r^{t}}\cdot \ell_{r^t}\label{eq_3_proof_theo_symm_and_linear_1}\\
a_{\bar{r}^t}\cdot M+B/2&>K\cdot (a_1\cdot(M+1)+B) \label{eq_4_proof_theo_symm_and_linear_1}
\end{align}
Using that~\eqref{eq_4_proof_theo_symm_and_linear_1} is equivalent to $K a_{\bar{r}^t}\cdot M>K^2a_1\cdot(M+1)+(K^2-K/2)B$
and combining this with~\eqref{eq_2_proof_theo_symm_and_linear_1} and~\eqref{eq_3_proof_theo_symm_and_linear_1} yields
\begin{align*}
K^2a_1\cdot(M+1)+(K^2-K/2)B&<a_{r^{t}}\cdot \ell_{r^t}\\ &<\frac{a_1\cdot(M+1)+B}{K}.
\end{align*}
From this we conclude the contradiction
\[
0\leq (K^3-1)a_1\cdot(M+1)<(-K^3+K^2/2+1)B=0,
\]
where we additionally used that $K^3>1$ and $-K^3+K^2/2+1=0$ by definition of $K$. 
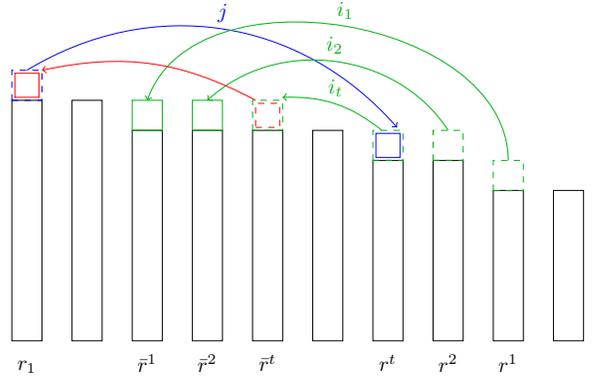
\begin{figure}[h]
\centering
\scalebox{0.8}{
\begin{tikzpicture}[decoration=brace]
\draw  (1, 0) rectangle (1.5, 4);
\draw[blue,dashed] (1, 4) rectangle (1.5, 4.5); 
\draw[red] (1.05, 4.05) rectangle (1.45, 4.45);
\node[below] at (1.25, -0.2) {$r_1$};
\path[->,blue] 
(1.25,4.5) edge[bend left=40]
node[above]{$j$} (7.4,3.55);
\draw  (2, 0) rectangle (2.5, 4); 
\draw  (3, 0) rectangle (3.5, 3.5); 
\draw[blue!20!black!30!green] (3, 3.5) rectangle (3.5, 4); 
\node[below] at (3.25, -0.1) {$\bar{r}^1$};
\draw  (4, 0) rectangle (4.5, 3.5); 
\draw[blue!20!black!30!green] (4, 3.5) rectangle (4.5, 4); 
\node[below] at (4.25, -0.1) {$\bar{r}^2$};
\draw  (5, 0) rectangle (5.5, 3.5); 
\draw[blue!20!black!30!green,dashed] (5, 3.5) rectangle (5.5, 4); 
\draw[red,dashed] (5.05, 3.55) rectangle (5.45, 3.95); 
\node[below] at (5.25, -0.1) {$\bar{r}^t$};
\draw  (6, 0) rectangle (6.5, 3.5); 
\draw  (7, 0) rectangle (7.5,3); 
\draw[blue!20!black!30!green,dashed] (7, 3) rectangle (7.5, 3.5); 
\draw[blue] (7.05, 3.05) rectangle (7.45, 3.45); 
\node[below] at (7.25, -0.1) {$r^t$};
\path[->,blue!20!black!30!green] 
(9.25,3) edge[bend right=80]
node[above]{$i_1$} (3.25,4);
\draw  (8, 0) rectangle (8.5, 3); 
\draw[blue!20!black!30!green,dashed] (8, 3) rectangle (8.5, 3.5); 
\node[below] at (8.25, -0.1) {$r^2$};
\path[->,blue!20!black!30!green] 
(8.25,3.5) edge[bend right=50]
node[above]{$i_2$} (4.25,4);
\draw  (9, 0) rectangle (9.5, 2.5); 
\draw[blue!20!black!30!green,dashed] (9, 2.5) rectangle (9.5, 3); 
\node[below] at (9.25, -0.1) {$r^1$};
\draw  (10, 0) rectangle (10.5, 2.5); 
\path[->,blue!20!black!30!green] 
(7.15,3.5) edge[bend right=20]
node[above]{$i_t$} (5.5,4.05);
\path[->,red]
(5.05,4) edge[bend right=20]
(1.5,4.5);
\end{tikzpicture}
}
\caption{Illustration for the case $\bar{r}^t \rightarrow r_1$ (with $t=3$).}%
\label{fig_4}%
\end{figure}

Therefore, the only remaining case is that a player on $\bar{r}^t$ changes to a resource $r''$ with load $\leq M-2$. 
We will first show that $r''=r^{t-1}$ holds. 
Note that the cost for deviating to $r''$ is strictly smaller than the cost for deviating to $r_1$ (otherwise, one would change to $r_1$). Therefore, $r''$ needs to have strictly smaller load now than it had before player~$k$ was added (otherwise, player~$k$ would be placed on $r''$ rather than on $r_1$). 
This shows that $r''\in \{r^1,\ldots,r^{t-1}\}$ needs to hold, and since deviating to $r^{t-1}$ is cheapest and $r^{t-1}$ has the smallest index among the cheapest deviations, the assertion follows.
Consider the situation after the change from $\bar{r}^t$ to $r^{t-1}$ (of some player~$j'$). It may be that a player on $\bar{r}^{t-1}$ (biggest index among the resources with load~$M$) wants to change to $r_1$ (the best resource with load $M$) or to a resource with load~$\leq M-2$. Players with load~$\leq M-1$ do not want to change, since a deviation to $\bar{r}^t$ (the best resource with load~$M-1$) cannot be beneficial (player~$j'$ moved first).  
Furthermore, the case $\bar{r}^{t-1} \rightarrow r_1$ leads to a contradiction (note that in this case, \eqref{eq_2_proof_theo_symm_and_linear_1}-\eqref{eq_4_proof_theo_symm_and_linear_1} hold since $a_{\bar{r}^t}\geq a_{\bar{r}^{t-1}}$). 
Thus we are in the situation that a player on $\bar{r}^{t-1}$ changes to a resource with load~$\leq M-2$. 
By repeating the argumentation of the last paragraph, we conclude that there can be no more than $t-3$ further deviations.  
\end{itemize}
\end{proof}

\begin{proof}[Proof of Corollary~\ref{cor_2}]
Note that the only case where the proof of Theorem~\ref{theo_symm_and_linear_1} fails if we consider $\alpha=1$ instead of $\alpha=K$ is the case corresponding to Figure~\ref{fig_4}. Since $r_1,\bar{r}^t$ and $r^t$ are three different resources, and with respect to $x'$ there are at least $3$ players using $r_1$, at least $1$ player using $\bar r^t$ and at least 1 player using $r^t$, the statement follows.
\end{proof}

\section{Omitted Proofs from Section~\ref{sec:optimal}}\label{appendix_sec4}
 \begin{proof}[Proof of Lemma~\ref{lemma_existence}]
Given $\bar{c}_{M}\geq 0$ and $\bar{c}_{< M}\geq 0$, the following procedure determines a vector $\ell=\ell(x)$ corresponding to an $\alpha$-PNE $x$ with $M(x)=M$, $k(x)=k$, $k'(x)=k'$ and $k''(x)=k''$, as well as $\bar{c}_M(x)\geq \bar{c}_M$ and $\bar{c}_{< M}(x)\geq\bar{c}_{< M}$, if such a strategy profile $x$ with the additional property that $\bar{c}_M(x)= \bar{c}_M$ and $\bar{c}_{< M}(x)=\bar{c}_{< M}$  exists.
\begin{enumerate}
	\item Test first whether all of the following inequalities are satisfied (if at least one inequality is not satisfied, there is no $\alpha$-PNE with the desired properties and we can stop):
	\begin{align*}
	a_k\cdot M+B/k &\leq \alpha\cdot\bar{c}_M \\
	\text{if $k'\geq k+2$: } a_{k'-1}\cdot(M-1)&\leq \alpha\cdot\bar{c}_{< M} \\
	\text{if $k'<k''$: } a_{k''-1}\cdot (M-2) & \leq \alpha\cdot\bar{c}_{< M} \\
	\text{if $k\geq 2$: } a_1\cdot (M+1)+B & \geq \bar{c}_M \\
	\text{if $k'\geq k+2$: } a_{k+1}\cdot M+\frac{B}{k} & \geq \bar{c}_M \\
	\text{if $k=1$ and $k'<k''$: } a_{k'}\cdot (M-1)+\frac{B}{k'} &\geq \bar{c}_M \\
	\text{if $k\geq 2$ and $k'<k''$: } a_{k'}\cdot (M-1) &\geq \bar{c}_M \\
	a_1\cdot(M+1)+B & \geq \bar{c}_{< M} \\
	\text{if $k'\geq k+2$: } a_{k+1}\cdot M +\frac{B}{k+1} & \geq \bar{c}_{< M}  \\
	\text{if $k'<k''$: } a_{k'}\cdot (M-1) &\geq \bar{c}_{< M}
	\end{align*}
	Note that the first three inequalities come from the Nash conditions (players using resources with load $M,M-1$ or $M-2$ do not want to deviate). The remaining inequalities are due to the minimum properties of $\bar{c}_M$ and $\bar{c}_{< M}$. 
	\item Set $\ell_r := M$ $\forall$ $r\in \{1,\ldots,k\}$, 
$\ell_r := M-1$ $\forall$ $r\in \{k+1,\ldots,k'-1\}$, 
$\ell_r := M-2$ $\forall$ $r\in \{k',\ldots,k''-1\}$. 

If $\ell_r<0$ for some of the above assigned values, we can stop since there is no $\alpha$-PNE with the desired properties. 
	\item Let $n':=n-k\cdot M-(k'-k-1)\cdot (M-1)-(k''-k')\cdot (M-2)$ be the number of players which are not assigned yet. 
	\item\label{eq:load} Test for all $r\geq k''$ whether the following inequalities lead to a contradiction (if this is the case, there is no $\alpha$-PNE with the desired properties and we can stop): 
	\begin{align*}
	\ell_r & \geq 0 \\
	\ell_r&\leq M-3 \\
		\ell_r& \leq \lfloor\frac{\alpha \cdot\bar{c}_{< M}}{a_r}\rfloor \\
		\ell_r& \geq \lceil\frac{\bar{c}_{< M}}{a_r}\rceil-1 \\
	\ell_r & \geq \lceil\frac{\bar{c}_{M}}{a_r}\rceil-1 
	\end{align*}
	Note that the third inequality ensures that no player on resource $r$ wants to deviate. The fourth and fifth inequality are due to the minimum properties of $\bar{c}_M$ and $\bar{c}_{< M}$.
	\item For all $r\geq k''$, let $b_r$ and $b_r'$ be the lower and upper bounds on $\ell_r$ which are induced by the inequalities in \ref{eq:load}., that is, $b_r:=\max\{0,\lceil\frac{\bar{c}_{< M}}{a_r}\rceil-1,\lceil\frac{\bar{c}_{M}}{a_r}\rceil-1\}$ and $b_r':=\min\{M-3,\lfloor\frac{\alpha \cdot\bar{c}_{< M}}{a_r}\rfloor\}$.
	\item If $n'\notin [\sum_{r\geq k''}{b_r}, \sum_{r\geq k''}{b_r'}]$, there is no $\alpha$-PNE with the desired properties and we can stop. 
	\item Else set $\ell_r:=b_r$ for all $r\geq k''$ and update $n'\leftarrow n'-\sum_{r\geq k''}{b_r}$. 
	\item If $n'>0$, then consider the resources $r\geq k''$ from $k''$ to $m$ and update $\ell_r\leftarrow \min\{b_r',b_r+n'\}$ and $n'\leftarrow n'-(\ell_r-b_r)$, until $n'=0$. 
\end{enumerate}
It is clear that if an $\alpha$-PNE $x$ with $M(x)=M$, $k(x)=k$, $k'(x)=k'$ and $k''(x)=k''$, as well as $\bar{c}_M(x)= \bar{c}_M$ and $\bar{c}_{< M}(x)=\bar{c}_{< M}$ exists, the above procedure terminates with a vector $\ell$. We now show that $\ell$ has the desired properties. 
Clearly, $\ell_r\geq 0$ for all $r\in R$ and $\sum_{r\in R}{\ell_r}=n$. Thus there is a strategy profile $x$ with load vector $\ell(x)=\ell$. We now argue that $x$ is an $\alpha$-approximate PNE with the desired properties. 
Clearly, $M(x)=M$, $k(x)=k$, $k'(x)=k'$ and $k''(x)=k''$. Furthermore, $\bar{c}_M\leq \bar{c}_M(x)$ and $\bar{c}_{< M}\leq \bar{c}_{< M}(x)$. Finally, $a_k\cdot M+B/k \leq \alpha\bar{c}_M$ and $a_r(\ell_r(x))\leq \alpha\bar{c}_{< M}$ for all $r>k$. Thus $x$ is in fact an $\alpha$-PNE. 

Note that for a strategy profile $x$ with $M(x)=M$, $k(x)=k$, $k'(x)=k'$ and $k''(x)=k''$, there are at most $3+(m-k''+1)\cdot(M-2)_+=O(mn)$ many possible values for $\bar{c}_{< M}(x)$, and also at most $3+(m-k''+1)\cdot(M-2)_+=O(mn)$ many possible values for $\bar{c}_{M}(x)$. 
Thus if we apply the above procedure for all possible choices of $\bar{c}_{M}(x)$ and $\bar{c}_{< M}(x)$, we get an $\alpha$-PNE~$x$ with $M(x)=M$, $k(x)=k$, $k'(x)=k'$ and $k''(x)=k''$, if such an $\alpha$-PNE exists. 
Since the number of times that we need to apply the above procedure (for given $M,k,k',k''$) is bounded by $O(n^2m^2)$, and the procedure itself is efficient, the overall algorithm is efficient, too.  
\end{proof}

\section{Existence of an approximate PNE with decreasing load profile}
\begin{applemma}\label{lemma_decreasingload}
For a multi-leader congestion game with an adversary where the resource set $R=[m]$ is ordered such that $a_1\leq \cdots \leq a_m$, the following holds for any $\alpha\in [1,2]$: 
If there exists an $\alpha$-approximate PNE, then there also exists an $\alpha$-approximate PNE~$x$ with decreasing loads, that is, with $\ell_1(x)\geq \cdots \geq \ell_m(x)$.
\end{applemma}
\begin{proof}
Assume there is a $\alpha$-approximate Nash equilibrium $y$ with $a_r \geq a_{r'}$ and $\ell_r(y)>\ell_{r'}(y)$ for some $r,r' \in R$. Then, we show that the profile $x$ with $\ell_r(x)=\ell_{r'}(y)$, $\ell_{r'}(x)=\ell_{r}(y)$ and $\ell_{\bar{r}}(x)=\ell_{\bar{r}}(y)$ for all $\bar{r} \in R\setminus\{r,r'\}$ is a  $\alpha$-approximate Nash equilibrium too. The profile $x$ results from $y$ by exchanging all players from $r$ and $r'$. Denote $z_{u,v}$ for the profile which results from a profile $z \in X$ by moving one player from a resource $u \in R$ to a resource $v \in R$. Since $y$ is a $\alpha$-approximate Nash equilibrium we have:
\begin{equation}\label{ynash}
a_r\ell_r(y)+\kappa_r^*(y) \leq \alpha(a_{\tilde{r}}(\ell_{\tilde{r}}(y)+1)+\kappa_{\tilde{r}}^*(y_{r,\tilde{r}}))
\end{equation}
for all $\tilde{r} \in R \setminus\{r\} $. First, we show, that a player on the resource $r$ is satisfied with his strategy considering the profile $x$. Let $M:=M(y)=\max\{\ell_{\tilde{r}}(y): \tilde{r}\in R\}$ denotes the maximum load with respect to $y$. If $\ell_r(y)=M$, $\ell_{r'}(y) \leq M-1$ and $\ell_{\bar{r}}(y)=M-1$ for a resource $\bar{r} \in R \setminus\{r,r'\}$ or $\ell_r(y)=M$, $\ell_{r'}(y) \leq M-1$ and $\ell_{\bar{r}}(y)=M-2$ for a resource $\bar{r} \in R \setminus\{r,r'\}$ and $r$ is the only resource with load $M$ regarding $y$ we know with Lemma \ref{lemma_deviation} that the cost of a deviation from $r$ to $\bar{r}$ must be not the same than the costs of a deviation from $r'$ to $\bar{r}$. In all the other cases the cost are equal and with inequality \eqref{ynash} we get:
\begin{align*}
c_r(x)&=a_r\ell_r(x)+\kappa_r^*(x)\\
&\leq a_r\ell_r(y)\\
&\leq \alpha(a_{\bar{r}}(\ell_{\bar{r}}(y)+1)+\kappa_{\bar{r}}^*(y_{r,\bar{r}}))\\
&=\alpha(a_{\bar{r}}(\ell_{\bar{r}}(x)+1)+\kappa_{\bar{r}}^*(x_{r,\bar{r}})),
\end{align*}
for all $\bar{r} \in R \setminus\{r',r\}$. Note that $\ell_r(x)=\ell_{r'}(y)<\ell_r(y)=\ell_{r'}(x)$ and $\kappa_{r}^*(x)=0$ since $r$ has not maximal load regarding $x$. Now, let us discuss the other two cases. First let $\ell_r(y)=M$, $\ell_{r'}(y) \leq M-1$ and $\ell_{\bar{r}}(y)=M-1$ for a resource $\bar{r} \in R\setminus\{r,r'\}$. We prove by contradiction, that in this situation a player on $r$ does not want to deviate to $\bar{r}$ considering the profile $x$. Let us assume that a player on $r$ wants to deviate to $\bar{r}$ considering the profile $x$, that means:
\begin{equation}\label{dev}
c_r(x)=a_r\ell_r(x) > \alpha(a_{\bar{r}}M+\frac{B}{q}),
\end{equation}
where $q:=|M^{-1}(y)| \geq 1$ denotes the number of resources with load $M$ in $y$. Since $y$ is a Nash equilibrium we know:
\begin{equation}\label{ne}
c_r(y)=a_rM+\frac{B}{q} \leq \alpha(a_{\bar{r}}M+\frac{B}{q+1},
\end{equation}
where, again, $q:=|M^{-1}(y)| \geq 1$ denotes the number of resources with load $M$ in $y$. With inequality \eqref{dev} and \eqref{ne} we get:
\begin{equation*}
\frac{a_rM+\frac{B}{q}}{\alpha}-\frac{B}{q} < \frac{a_r(M-1)+\frac{B}{q}}{\alpha}-\frac{B}{q+1}.
\end{equation*}
Thus, we have
\begin{equation*}
B(q+1-\alpha) < -a_rq(q+1) \leq 0.
\end{equation*}
But with $B>0$ and $q+1-\alpha \geq 1+1-2=0$ we get a contradiction. 
If $\ell_r(y)=M$, $\ell_{r'}(y) \leq M-1$ and $\ell_{\bar{r}}(y)=M-2$ for a resource $\bar{r} \in R \setminus\{r,r'\}$ and $r$ is the only resource with load $M$ regarding $y$ we also can show by contradiction, that in this situation a player on $r$ does not want to deviate to $\bar{r}$ considering the profile $x$. Let us assume that a player on $r$ wants to deviate to $\bar{r}$ considering the profile $x$, that means:
\begin{equation}\label{dev2}
c_r(x)=a_r\ell_r(x) > \alpha(a_{\bar{r}}(M-1)).
\end{equation}
Since $y$ is a Nash equilibrium we know:
\begin{equation}\label{ne2}
c_r(y)=a_rM+B \leq \alpha(a_{\bar{r}}(M-1)+\frac{B}{p+2}),
\end{equation}
where $p$ denotes the denotes the number of resources with load $M-1$ in $y$. With inequality \eqref{dev2} and \eqref{ne2} we get:
\begin{equation*}
\frac{a_rM+B}{\alpha}-\frac{B}{p+2}+a_{\bar{r}} < \frac{a_r(M-1)}{\alpha}+a_{\bar{r}},
\end{equation*}
where it follows
\begin{equation*}
B(1-\frac{\alpha}{p+2}) < -a_r \leq 0.
\end{equation*}
But with $B>0$ and $1-\frac{\alpha}{p+2}\geq 1-\frac{2}{2}=0$ we get a contradiction. Altogether, 
we know that a player on the resource $r$ does not want to change to a resource $\bar{r} \in R\setminus\{r,r'\}$ considering the profile $x$.
It remains to show that no player on $r$ wants to deviate to $r'$. With inequality \eqref{ynash} we get:
\begin{align*}
c_r(x)&=a_r\ell_r(x)+\kappa_r^*(x)\\
&\leq a_r\ell_r(y)\\
&\leq \alpha(a_{r'}(\ell_{r'}(y)+1)+\kappa_{r'}^*(y_{r,r'}))\\
&\leq \alpha(a_{r'}(\ell_{r'}(x)+1)+\kappa_{r'}^*(x_{r,r'})).
\end{align*}
Thus, the players on $r$ are satisfied with their strategy considering $x$. Now, let us prove the same for the players on $r'$. Since $a_r \geq a_{r'}$ we have:
\begin{align*}
c_{r'}(x)&=a_{r'}\ell_{r'}(x)+\kappa_{r'}^*(x)\\ &\leq a_r\ell_r(y)+\kappa_r^*(y)\\
&\leq \alpha(a_{\bar{r}}(\ell_{\bar{r}}(y)+1)+\kappa_{\bar{r}}^*(y_{r,\bar{r}}))\\
&\leq \alpha(a_{\bar{r}}(\ell_{\bar{r}}(x)+1)+\kappa_{\bar{r}}^*(x_{r',\bar{r}}),
\end{align*}
for all $\bar{r} \in R \setminus\{r',r\}$. The second inequality holds due to inequality \eqref{ynash}. Furthermore with inequality \eqref{ynash} and $a_r \geq a_{r'}$ a player on $r'$ does not want to change to $r$:
\begin{align*}
c_{r'}(x)&=a_{r'}\ell_{r'}(x)+\kappa_{r'}^*(x)\\ &\leq a_r\ell_r(y)+\kappa_r^*(y)\\
&\leq \alpha(a_{r'}(\ell_{r'}(y)+1)+\kappa_{r'}^*(y_{r,r'}))\\
&\leq \alpha(a_{r}(\ell_{r}(x)+1)+\kappa_{r}^*(x_{r',r})).
\end{align*}
Thus, the players on $r'$ do not want to deviate considering the profile $x$. Finally, let us show that the players on $\bar{r} \in R \setminus\{r',r\}$ are satisfied with their strategy regarding $x$. Clearly, they do not want to deviate to a resource which is not $r,r'$ since they have the same costs regarding $x$ and $y$. Furthermore we know, that $y$ is a $\alpha$-approximate Nash equilibrium and a deviation to an resource which is not $r,r'$ would provide the same cost considering $x$ and $y$. Thus, it remains to show that a player on $\bar{r}$ does not want to change to $r$ and $r'$. First, let us consider a deviation from $\bar{r}$ to $r$. Since $a_r \geq a_{r'}$ and $y$ is a $\alpha$-approximate Nash equilibrium we have:
\begin{align*}
c_{\bar{r}}(x)&=a_{\bar{r}}\ell_{\bar{r}}(x)+\kappa_{\bar{r}}^*(x)\\ &=a_{\bar{r}}\ell_{\bar{r}}(y)+\kappa_{\bar{r}}^*(y)\\
&\leq \alpha(a_{r'}(\ell_{r'}(y)+1)+\kappa_{r'}^*(y_{\bar{r},r'}))\\
&\leq \alpha(a_{r}(\ell_{r}(x)+1)+\kappa_{r}^*(x_{\bar{r},r})).
\end{align*}
To complete the proof let us consider a deviation from $\bar{r}$ to $r'$. Since $a_r \geq a_{r'}$ and $y$ is a $\alpha$-approximate Nash equilibrium we get:
\begin{align*}
c_{\bar{r}}(x)&=a_{\bar{r}}\ell_{\bar{r}}(x)+\kappa_{\bar{r}}^*(x)\\ &=a_{\bar{r}}\ell_{\bar{r}}(y)+\kappa_{\bar{r}}^*(y)\\
&\leq \alpha(a_{r'}(\ell_{r'}(y)+1)+\kappa_{r'}^*(y_{\bar{r},r'}))\\
&\leq \alpha(a_{r'}(\ell_{r'}(x)+1)+\kappa_{r'}^*(x_{\bar{r},r'})).
\end{align*}
Altogether, we can conclude that $x$ is a $\alpha$-approximate Nash equilibrium.
\end{proof}

\section{Additively Approximate PNE}
\label{app:additive-approximate}
In addition to multiplicatively approximate equilibria, as studied in this paper, it also possible to define \emph{additively approximate equilibria}. A \emph{$\varepsilon$-additive approximate PNE} is a strategy profile $x \in X$ such that
\[ \pi_i(x)\leq \pi_i(y_i,x_{-i}) + \varepsilon \text{ for all }y_i\in X_i.\]
Contrasting the case of $K$-multiplicative approximate PNE, however, the existence of $\varepsilon$-additive approximate PNE for the class of games studied in this paper cannot be guaranteed for any constant $\varepsilon > 0$.

To see this, note that the considered multi-leader congestion games are invariant under scaling in the sense that multiplying all cost coefficients $a_r$ and the adversary's budget $B$ with the same factor $\lambda \geq 0$ results in a game in which each player's private cost for any given strategy profile is scaled by the same factor $\lambda$. As a result, given any constant $\varepsilon > 0$, any instance of a game that does not have an exact PNE (such as, e.g., the one described in 
Section~\ref{sec:tightness}) can be scaled in such a way that it does not allow for a $\varepsilon$-additive approximate PNE.

We remark that, while the existence of approximate PNE with a small additive constant cannot be guaranteed, the approach in Section~\ref{sec:optimal} can easily be adjusted to compute $\varepsilon$-additive approximate PNE with the smallest possible $\varepsilon$ for a given instance. 

\end{document}